\documentclass[journal,10pt]{IEEEtran} 
\makeatletter
\def\endthebibliography{%
	\def\@noitemerr{\@latex@warning{Empty `thebibliography' environment}}%
	\endlist
}
\makeatother

\usepackage{cite}

\usepackage[pdftex]{graphicx}

\usepackage[caption=false,font=footnotesize]{subfig}

\usepackage{amsmath}
\usepackage{mathtools, cuted}
\usepackage{pifont}
\usepackage{filecontents}
\usepackage{amssymb}
\usepackage{color}
\usepackage{xcolor}

\usepackage{epsfig}

\usepackage{url}

\usepackage{algpseudocode}
\usepackage{algorithm, tabularx}

\usepackage{lettrine}

\usepackage{lipsum}

\usepackage{siunitx}

\usepackage{soul,color}

\usepackage{mathrsfs}

\usepackage{array}

\usepackage[inline]{enumitem}

\usepackage{breqn}

\usepackage{multirow}
\usepackage{array}
\newcolumntype{L}[1]{>{\raggedright\let\newline\\\arraybackslash\hspace{0pt}}m{#1}}
\newcolumntype{C}[1]{>{\centering\let\newline\\\arraybackslash\hspace{0pt}}m{#1}}
\newcolumntype{R}[1]{>{\raggedleft\let\newline\\\arraybackslash\hspace{0pt}}m{#1}}

\usepackage{amsthm}
\newtheorem{theorem}{Theorem}

\newlength{\maxwidth}
\newcommand{\algalign}[2]
{\makebox[\maxwidth][r]{$#1{}$}${}#2$}

\makeatletter
\newcommand{\multiline}[1]{%
	\begin{tabularx}{\dimexpr\linewidth-\ALG@thistlm}[t]{@{}X@{}}
		#1
	\end{tabularx}
}
\makeatother

\DeclareMathOperator*{\argmax}{\arg\max}
\theoremstyle{remark}

\newtheorem{remark}{Remark}

\algdef{SE}[SUBALG]{Indent}{EndIndent}{}{\algorithmicend\ }%
\algtext*{Indent}
\algtext*{EndIndent}

\newcommand{\cmark}{\ding{51}}%

\begin{document}
	
	\title{Performance Analysis and Optimization of FAS-ARIS Communications for 6G:\\System Modeling and Analytical Insights}
	
	\author{Hong-Bae Jeon,~\IEEEmembership{Member,~IEEE}, Kai-Kit Wong,~\IEEEmembership{Fellow,~IEEE}, and Chan-Byoung Chae,~\IEEEmembership{Fellow,~IEEE}
		
		\thanks{This work was supported by Hankuk University of Foreign Studies Research Fund of 2026. \textit{(Corresponding Author: Chan-Byoung Chae.)}}
		\thanks{H.-B. Jeon is with the Department of Information Communications Engineering, Hankuk University of Foreign Studies, Yong-in 17035, South Korea (e-mail: hongbae08@hufs.ac.kr).}
		\thanks{K.-K. Wong is with the Department of Electronic and Electrical Engineering, University College London, WC1E 6BT London, U.K., and also with the Yonsei Frontier Laboratory, Yonsei University, Seoul 03722, South Korea (e-mail: kai-kit.wong@ucl.ac.uk).}
		\thanks{C.-B. Chae is with the School of Integrated Technology, Yonsei University, Seoul 03722, South Korea (e-mail: cbchae@yonsei.ac.kr).}
		}

	\maketitle

	\begin{abstract}
This paper introduces a unified analytical and optimization framework for fluid antenna system-active reconfigurable intelligent surface (FAS-ARIS) communications in 6G. By combining the port reconfigurability of FAS with the signal amplification of ARIS, the proposed design enables more flexible control of the propagation environment and enhanced link reliability beyond what passive solutions can offer. We first derive the optimal ARIS amplification gain under a reflection power constraint to maximize the user's signal-to-noise ratio (SNR). Using a block-diagonal matrix approximation, we obtain a tractable outage expression and a tight independent-antenna equivalent upper-bound. Building on this, we establish the monotonic relationship between outage and effective channel gain, which enables a closed-form solution for ARIS phase optimization under limited channel state information (CSI). To further improve spectral efficiency, we propose a region-partitioned throughput optimization framework that achieves near-optimal performance without exhaustive search, thereby verifying its low computational complexity. Extensive simulations confirm the accuracy of the analysis and demonstrate consistent gains in outage and throughput compared to baselines.

	\end{abstract}


	\IEEEpeerreviewmaketitle
	
	\section{Introduction}
\lettrine{O}{ver} the past few years, the emergence of sixth-generation (6G) wireless networks has promised major improvements in connectivity, latency, and energy- and spectral-efficiency~\cite{seman, HBFSO, mapxmag, hyoo}. At the core of these advancements lies multiple-input multiple-output (MIMO) technology, which uses antenna arrays to exploit spatial diversity and multiplexing gains~\cite{trimimo, trimimo22}. However, deploying MIMO in compact devices such as smartphones and Internet-of-things (IoT) devices is challenging, because conventional MIMO requires antenna spacing of at least half a wavelength to achieve low correlation between antennas. To address this limitation, the fluid antenna system (FAS) (also known as a movable antenna) has been proposed as a promising alternative~\cite{FAS, fluidtut, movetut, movemag}. FAS enables the use of spatial degrees of freedom (DoF) without requiring multiple physically separated antennas by allowing a single antenna element to switch among multiple candidate positions within a predefined region~\cite{FASJSAC}. By selecting the port with the strongest instantaneous signal, an $N$-port FAS can achieve diversity gains comparable to those of an $N$-antenna MIMO system, while using only one radio-frequency (RF) chain~\cite{fasopdg}. Practical prototypes, including liquid-metal~\cite{liquid} and pixel-based reconfigurable antennas~\cite{pixel}, have demonstrated the feasibility of FAS and its strong potential to deliver significant diversity gains and robust mitigation against small-scale fading~\cite{FAScor, fasopdg, newcor}.

Building on its promising capabilities, FAS has been integrated into various wireless paradigms. In~\cite{FASqout}, the authors highlight that MIMO-FAS significantly outperforms conventional MIMO systems by leveraging reconfigurable port selection and spatial diversity by the proposed ``$q$-outage'' capacity metric. The work in~\cite{newcor} introduces a spatial block-correlation model as a tractable yet accurate approximation for arbitrary 1D or 2D FAS correlation structures, enabling efficient analysis and simulation by alleviating the complexity of realistic models such as Clarke's or Jakes'. A common insight shared by recent FAS studies is that even a compact single-antenna architecture can exploit spatial diversity to achieve performance comparable to multi-antenna systems. Both~\cite{fasopdg} and~\cite{perlimfas} show that, by intelligently leveraging port reconfigurability, a single-RF-chain FAS can approximate the capacity and outage performance of conventional maximum-ratio combining (MRC) receivers, thereby surpassing the limitations of traditional single-antenna designs. The work in~\cite{fasover} shows that while spatial oversampling can greatly enhance the diversity and outage performance of FAS, and to mitigate substantial switching overhead, the authors propose a low-complexity port selection strategy that leverages spatial correlation to reduce switching while maintaining near-optimal performance. Recent studies also demonstrate that machine-learning (ML) techniques are being actively explored in FAS, with deep reinforcement learning applied to joint port selection and precoder design for multiuser MIMO with integrated-sensing-and-communications (ISAC)~\cite{fasrl}, and model-free multi-agent learning frameworks utilized for opportunistic scheduling and optimal port selection in dynamic environments~\cite{fasal}. Collectively, these works validate the potential of FAS in enhancing coverage and reliability in compact wireless terminals.

{While FAS effectively addresses spatial limitations at the receiver by exploiting port reconfigurability, another major technological breakthrough in 6G is the reconfigurable intelligent surface (RIS), which enables environmental control of radio propagation. By adjusting the phase shifts (and potentially the amplitude) of a large number of reflecting elements~\cite{RISvtm, holotut}, RIS can steer incident signals toward the receiver, create artificial line-of-sight (LoS) paths, and mitigate blockage or harsh propagation conditions~\cite{dsRIS22}. However, conventional passive-RIS (PRIS) suffers from the well-known double-fading effect, where the end-to-end channel experiences the product of two path losses (e.g., base-station (BS)-RIS and RIS-user), which fundamentally limits its effectiveness in long-range or high-frequency scenarios. To overcome this limitation, active-RIS (ARIS) has been proposed~\cite{aristut}, where amplification circuits are embedded into RIS elements to actively boost the reflected signal. While ARIS significantly improves SNR and coverage, it also introduces new challenges, including amplification noise, reflection power constraints, and increased optimization complexity~\cite{aris1, aris5, aris4}.}

{Motivated by the complementary capabilities of FAS and RIS, their integration has recently attracted growing research interest. In particular, combining FAS with PRIS (FAS-PRIS) allows the environment to shape the propagation channel while the receiver exploits spatial diversity via port selection. In~\cite{FASRIS, FASRIS22}, the authors proposed a FAS-PRIS framework using a block-diagonal correlation model to analyze outage probability and throughput in both CSI-aware and CSI-free settings~\cite{FASRIS22}, demonstrating significant performance benefits. The authors in~\cite{fasrisper} examined average throughput in a FAS-PRIS system and emphasized the importance of FAS correlation. 
The authors in~\cite{powerFASRIS} investigate transmit power minimization under quality-of-service (QoS) constraints in a multiuser FAS-PRIS system by jointly optimizing the BS precoder, PRIS elements, and FAS positions at the users. Despite these advances, FAS-PRIS systems remain inherently power-limited due to the passive nature of PRIS and therefore struggle to maintain reliable performance in severe path-loss environments.}

{From this perspective, integrating FAS with ARIS represents a natural and necessary evolution from FAS-PRIS architectures. By replacing PRIS with ARIS, the environment is no longer restricted to passive beam shaping but can actively strengthen the incident signal through amplification, thereby complementing the spatial diversity gain provided by FAS. At the same time, this integration introduces fundamentally new challenges that do not arise in FAS-only or FAS-PRIS systems, as the amplification noise and power constraints of ARIS interact with the spatial correlation among densely sampled FAS ports. This coupled signal-noise-correlation effect renders existing analytical frameworks inadequate and directly motivates the development of a dedicated performance analysis and optimization framework for FAS-ARIS systems.}

Research on FAS-ARIS integration, however, is still in its infancy, largely due to the increased system complexity by dynamic amplification and noise~\cite{aris1}. FAS-ARIS, however, is markedly different from the aforementioned FAS-PRIS systems, because the active surface not only enhances the cascaded channel but also injects amplification noise whose power scales with the reflection gain, resulting in a non-trivial coupling between the desired signal, noise, and FAS port correlation. This makes exact performance analysis substantially more difficult than in the passive case. The authors in~\cite{FASARIS22} investigate the outage probability of FAS-ARIS systems using a block-correlation model and derive a closed-form expression, demonstrating its accuracy and the system's performance advantage for reliable 6G communications. 

In~\cite{FASARIS}, the authors investigate the relative importance of FAS and ARIS by developing an optimization-based system model and show that their significance depends on system complexity, with their integration yielding enhanced robustness and efficiency. In~\cite{fasaristnse}, the authors introduced an FAS-ARIS framework that exploits multipath propagation for high-accuracy localization in non-LoS (NLoS) environments, demonstrating that controlled amplification and phase adjustment can turn detrimental channel scattering into useful spatial information. Despite recent progress, as shown in Table~\ref{tabcomp}, most existing studies remain confined to PRIS scenarios~\cite{FASRIS22, fasrisper, FASRIS, powerFASRIS}, which still inherently suffer from limited signal amplification capabilities and reduced adaptability to dynamic environments. Even in the limited research on FAS integrating with ARIS, these works often rely on overly simplified or idealized FAS correlation models~\cite{FASARIS, fasaristnse}, restrict their analysis to outage probability alone~\cite{FASARIS, FASARIS22}, or assume overly simplified CSI by Rayleigh Fading~\cite{FASARIS22} or ideal CSI availability~\cite{FASARIS}, resulting in impractical signaling overhead and limited scalability of the FAS-ARIS systems. As a result, a comprehensive performance analysis of FAS-integrated ARIS systems under realistic spatial correlation and partial CSI remains an open and critical research challenge.

To bridge these research gaps, in this paper, we present a comprehensive investigation into the performance and optimization of integrated FAS-ARIS communication systems for 6G networks. We develop a unified framework that accounts for spatial correlation among fluid antenna ports, models the dynamic noise amplification introduced by ARIS elements, and accommodates limited-CSI scenarios, an inherent challenge in FAS-enabled environments. Through analytical derivations and numerical evaluations, we provide valuable design guidelines and practical insights for the efficient deployment of FAS-ARIS systems. Unlike prior studies on FAS-PRIS or partial FAS-ARIS architectures, which mainly focus on empirical performance trends or limited configurations, this work develops a unified analytical framework that quantifies the outage probability and throughput of the integrated FAS-ARIS system. In addition, we derive tractable approximations that explicitly capture the interplay among the parameters, allowing quantitative evaluation of key system parameters that were not analytically addressed in earlier works. The main contributions of this work are summarized as follows:
\begin{table}[t]
\centering
\caption{Comparison for FAS-ARIS/PRIS Studies}
\label{tabcomp}
\begin{tabular}{|>{\centering}m{2.15cm}|>{\centering}m{0.6cm}|>{\centering}m{0.6cm}|>{\centering}m{0.6cm}|>{\centering}m{0.6cm}|>{\centering}m{0.6cm}|>{\centering}m{0.6cm}|}
\hline
\centering \textbf{Reference}& \cite{FASRIS22} &\cite{fasrisper}& \cite{FASARIS22} & \cite{FASARIS} &\cite{fasaristnse} & {\textbf{This\\Work}} 
\tabularnewline
\hline
\centering Amplification\\Modeling
& - 
& -
& \cmark
& \cmark
& \cmark
& \cmark
\tabularnewline
\hline
\centering {Practical\\Correlation Model}
& \cmark
& \cmark
& \cmark
& -
& -
& \cmark
\tabularnewline
\hline
\centering {Realistic CSI Assumption} 
& \cmark
& -
& -
& -
& \cmark
& \cmark
\tabularnewline
\hline
\centering {Comprehensive Scope of Analysis} 
& \cmark
& \cmark
& -
& -
& -
& \cmark
\tabularnewline
\hline
\end{tabular}
\end{table}
\begin{itemize}

\item {We develop a unified downlink FAS-ARIS communication framework that explicitly captures the {coupled effects} of (i) spatial correlation among densely sampled FAS ports, and (ii) dynamic noise amplification induced by ARIS circuitry. Unlike, as depicted in Table~\ref{tabcomp}, existing FAS-PRIS studies that neglect amplification noise, or recent FAS-ARIS works that rely on idealized or oversimplified correlation models, the proposed framework integrates a block-diagonal matrix approximation (BDMA)~\cite{newcor} for realistic FAS correlation modeling with a power-constrained ARIS amplification model. Within this joint framework, we analytically derive the optimal ARIS amplification gain that maximizes the instantaneous received signal-to-noise ratio (SNR), thereby revealing how ARIS power, amplification noise, and FAS correlation interact in determining link reliability.}

\item We analyze the outage probability of the proposed FAS-ARIS system by evaluating the distribution of the parameters, incorporating the spatial correlation of FAS ports modeled via the BDMA framework and the amplification gain and dynamic noise introduced by the ARIS. To enable comprehensive performance evaluation, we also introduce an analytically tractable upper-bound based on the independent antenna equivalent (IAE) model. This two-tier modeling approach ensures both accuracy and computational efficiency, as verified by extensive simulations under realistic FAS-ARIS-aided environments.

\item We formulate the outage probability minimization problem with respect to the ARIS phase shifts and show that the outage probability decreases monotonically with the effective channel gain. Leveraging this property, we derive an optimal ARIS phase configuration that maximizes constructive signal combining at the selected FAS port. This solution provides important insights for practical ARIS phase control under limited-CSI conditions.

\item To optimize throughput while avoiding exhaustive search, we propose a computationally efficient algorithm that leverages the IAE-based approximation of outage probability. By dividing the search space into three operating regions, concave, quasiconcave, and non-monotonic, we adopt region-specific optimization strategies, consequently showing that the resulting algorithm is computationally efficient.

\item We evaluate the proposed FAS-ARIS system through extensive Monte Carlo simulations across diverse settings. 
The results validate that the proposed framework accurately captures system performance. Additionally, the IAE-based upper-bound is shown to be a tight approximation across various configurations. Furthermore, the proposed throughput optimization algorithm achieves near-optimal performance without exhaustive search over the whole domain, offering a practical trade-off between accuracy and computational complexity. In all scenarios, the FAS-ARIS system consistently outperforms conventional configurations, underscoring its potential as a highly efficient and robust solution for 6G.
\end{itemize}

\begin{figure}[t]
	\begin{center}
		\includegraphics[width=0.6\columnwidth,keepaspectratio]%
		{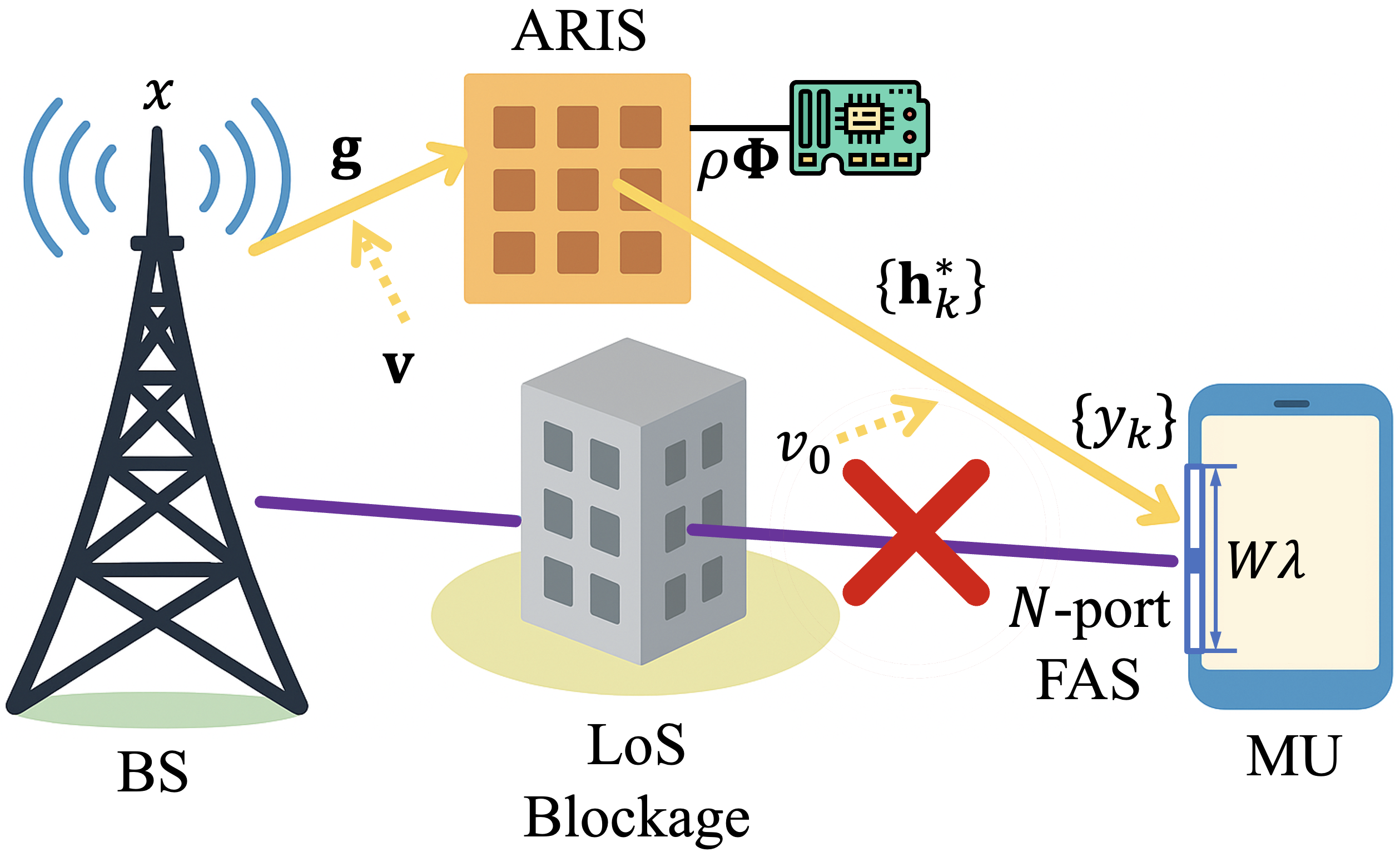}
		\caption{System model of the FAS-ARIS-aided wireless network.}
		\label{fig_overview}
	\end{center}
\end{figure}
\section{System Model}\label{4.1}
Consider a downlink FAS-ARIS system comprising a BS with a fixed-position antenna, an ARIS with $M$ reflecting elements, each with amplification gain $\rho$, and a mobile user (MU) equipped with FAS featuring $N$ ports, shown in Fig.~\ref{fig_overview}. The MU can dynamically select the most favorable port based on channel conditions across a linear region of length $W\lambda$, where $W$ represents the normalized distance relative to the carrier wavelength $\lambda$~\cite{FAS}. It is assumed that there is no direct link between the BS and the MU, making the ARIS essential for establishing and maintaining the communication link.

We assume that the ARIS is positioned near the BS~\cite{HBRIS, How, FASRIS22}, specifically when the BS-ARIS separation is less than the advanced MIMO Rayleigh distance (MIMO-ARD), 
thereby ensuring that the BS-ARIS channel is predominantly LoS~\cite{FASRIS}. Therefore, the BS-ARIS channel $\mathbf{g}\in\mathbb{C}^{M\times1}$ can be expressed as $\mathbf{g}=\sqrt{\beta}\bar{\mathbf{g}}$, where $\beta$ represents the path-loss between the BS and ARIS, and $\bar{\mathbf{g}}$ denotes the deterministic LoS component with each component has magnitude of 1 and $||\bar{\mathbf{g}}||_2=\sqrt{M}$~\cite{HBRIS}.
Meanwhile, due to the potentially long distance between the ARIS and the MU, the ARIS-MU link may include several NLoS components, which leads to the modeling of the ARIS-MU channel as a Rician fading model. Denoting the channel from the ARIS to the $k$-th port as $\mathbf{h}_k^*\in\mathbb{C}^{1\times M}$, it is given by
\begin{equation}
\label{RISu}
\mathbf{h}_k^* \triangleq \sqrt{\frac{\alpha K}{K+1}} \bar{\mathbf{h}}_k^* + \sqrt{\frac{\alpha}{K+1}} \tilde{\mathbf{h}}_k^*,
\end{equation}
where $\alpha$ is the path-loss between the ARIS and the MU, $K$ is the Rician factor, $\bar{\mathbf{h}}_k^*$ is the deterministic LoS component with $\lVert \bar{\mathbf{h}}_{k} \rVert_2^2 = M$, and $\tilde{\mathbf{h}}_k^*$ is the NLoS component drawn from $\mathcal{CN} \left(\mathbf{0},  \mathbf{I}_M \right)$, where $\mathbf{I}_M$ is an $M\times M$ identity matrix.
Here, we assume the far-field model, thus all ports in the FAS have the same LoS component, i.e., $\bar{\mathbf{h}}_k=\bar{\mathbf{h}}~(\forall k)$. {We assume that ARIS operates with only statistical CSI, as obtaining full instantaneous CSI in FAS-ARIS systems is highly challenging and incurs considerable overhead, especially in dynamic environments~\cite{FASRIS, FASRIS22, FASARIS22}. 
Note that this does not imply that MU lacks instantaneous channel knowledge. In the considered FAS-ARIS architecture, CSI availability is asymmetric across network entities~\cite{FASRIS, FASRIS22, FASARIS22}. The ARIS determines its phase configuration based on long-term channel statistics, i.e., the LoS components, which remain quasi-static over multiple transmission intervals, thereby avoiding frequent CSI acquisition and feedback overhead, which will be shown in Section~\ref{opop}. On the other hand, the MU performs fluid antenna port selection based on instantaneous SNR measurements across candidate ports, which will be shown in Section~\ref{oag}. This process does not require explicit estimation of the full cascaded BS-ARIS-MU channel, but relies solely on local observations at the receiver. Such a receiver-side adaptation mechanism is fully compatible with the assumption of statistical CSI at the ARIS.
}

In the case of FAS, the close spacing of the $N$ ports naturally leads to correlation among the channel vectors $\{\mathbf{h}_k^*\}_{k=1}^N$. According to Jakes' model, the correlation coefficient between the first and the $k$-th port is expressed as~\cite{FAScor}
\begin{equation}
\label{jakes}
\begin{aligned}
\mu_{1,k} \triangleq J_0 \left(\frac{2\pi(k-1)}{N-1}W\right)~(k \in \{ 1, \cdots, N\}\triangleq \mathbb{N}),
\end{aligned}
\end{equation}
where $J_0$ denotes the zeroth-order Bessel function of the first kind. The correlation matrix of the vector $\left[h_1^{(m)} \cdots h_N^{(m)}\right]^{\mathrm{T}}\in\mathbb{C}^N$, which consists of $m$-th entry of $\{\mathbf{h}_k\}_{k=1}^N$, where $m$ represents the $m$-th element of the ARIS, therefore forms a Toeplitz structure as described by~\cite{FAScor, newcor}
\begin{equation}
\label{newcorr}
\begin{aligned}
\mathbf{\Sigma} &\triangleq \textrm{topelitz}(\mu_{1,1}, \cdots, \mu_{1,N})\in\mathbb{R}^{N\times N}.
\end{aligned}
\end{equation}
To enable accurate yet simplified analysis, we adopt the BDMA model~\cite{newcor} to represent the spatial correlation structure. Herein, the spatial correlation among the $N$ ports is approximated using a block-diagonal correlation matrix:
\begin{equation}
\label{bdmahat}
\mathbf{\hat{\Sigma}} \triangleq \textrm{Blkdiag}(\mathbf{C}_1 , \cdots, \mathbf{C}_B) \in \mathbb{R}^{N \times N}.
\end{equation}
Each $\mathbf{C}_b$ is a constant correlation matrix of size $L_b$, with all off-diagonal entries set to the correlation coefficient $\mu_b^2$:
\begin{equation}
\label{cbdef}
\mathbf{C}_b \triangleq \textrm{toeplitz}(1, \mu_b^2, \cdots, \mu_b^2) \in \mathbb{R}^{L_b \times L_b},
\end{equation}
where $\sum_{b=1}^B L_b =N$. Note that with letting $\mu_b=\mu~(\forall b)$ for simplicity, where $\mu$ is close to 1, we can obtain $B, \{L_b\}_{b=1}^B$, and $\mathbf{\hat{\Sigma}}$ in~\cite[Algorithm 1]{newcor}. Under this assumption, the block-diagonal model successfully captures the spectrum of the true correlation matrix~\cite{newcor, FASRIS, FASRIS22, FASARIS22} and is preferred over full-correlation models like Jakes' because the latter are mathematically intractable for large FAS arrays, while earlier simplified models (e.g., constant correlation~\cite{FAScor}) lack realism and can mislead performance conclusions. In other words, BDMA strikes a principled trade-off between accuracy and tractability, by capturing the dominant eigenvalues that determine performance, supported by spatial oversampling and eigenvalue decaying properties of FAS~\cite{newcor}.

Within each submatrix $\mathbf{C}_b$ of size $L_b$, the channels $\{\mathbf{h}_k\}_{k=1}^{L_b}$ are mutually correlated. Based on the expressions in~(\ref{RISu}) and~(\ref{cbdef}), the channel vector $\mathbf{h}_k$ is modeled as~\cite{fama}
\begin{equation}
\label{parcor}
\mathbf{h}_k = \sqrt{\frac{\alpha K}{K+1}} \bar{\mathbf{h}} + \mu \tilde{\mathbf{h}}_b + \sqrt{1 - \mu^2} \mathbf{e}_k,
\end{equation}
where $\tilde{\mathbf{h}}_b$ and $\mathbf{e}_k$ are independent random vectors drawn from $\mathcal{CN} \left(\mathbf{0}, \frac{\alpha}{K+1} \mathbf{I}_M \right)$. For ARIS, we define the reflection and amplification matrix as $\rho \mathbf{\Phi} \triangleq \rho  \textrm{diag}\left(\{e^{j\theta_m}\}_{m=1}^M\right)$~\cite{aris4}, where each $\theta_m \in [0, 2\pi)$ denotes the phase shift introduced by the $m$-th ARIS element. When BS transmits a signal $x$ satisfying $\mathbb{E}[|x|^2] = 1$ to the MU through the ARIS, the signal received at the $k$-th port of MU is given by~\cite{aris5}
\begin{equation}
\label{aris}
y_k = \sqrt{P}\rho \mathbf{h}_k^* \mathbf{\Phi} \mathbf{g} x +\rho \mathbf{h}_k^* \mathbf{\Phi v} + v_0~(k \in \mathbb{N}),
\end{equation}
where $P$ represents the BS transmit power, $\mathbf{v}$ is the dynamic noise at the ARIS, and $v_0$ denotes the additive thermal noise at MU. In practice, $\mathbf{v}$ arises from hardware non-idealities such as gain instability, phase perturbations, and device-level thermal noise~\cite{aris1}, all of which distort the reflected signal and introduce excess noise power. Consequently, the analytical noise model of ARIS can be interpreted as characterizing the rate degradation imposed by realistic hardware imperfections, rather than ideal linear amplification. We assume $\mathbf{v} \sim \mathcal{CN}(\mathbf{0}, \sigma^2 \mathbf{I}_M)$ and $v_0 \sim \mathcal{CN}(0, \sigma_0^{2})$. Under a reflection power budget of $P'$, the power constraint is expressed as
\begin{equation}
\label{refl}
\begin{aligned}
&P \rho^2|| \mathbf{\Phi g} ||_2^2 +\rho^2 || \mathbf{\Phi v} ||_2^2 = P \rho^2 \beta M + \rho^2 M \sigma^2 \le P'.
\end{aligned}
\end{equation}
{\section{Problem Characterization}
\label{secpf}
We first formally define the performance metrics and optimization objectives considered in this paper.
\subsection{Outage Probability}
\label{secop}
Given a target transmission rate $R$, the outage probability at the MU is defined as
\begin{equation}
\label{outage}
P_{\mathrm{out}}(R) \triangleq \mathbb P\left[\log_2(1+\gamma) < R \right],
\end{equation}
where $\gamma$ denotes the received SNR after FAS port selection and ARIS-assisted transmission.
\subsection{Effective Throughput}
\label{secth}
To capture the trade-off between spectral efficiency and reliability, we define the effective throughput as
\begin{equation}
\label{Tdef}
T(R) \triangleq R\big(1-P_{\mathrm{out}}(R)\big),
\end{equation}
which represents the average successfully-delivered information rate.
\subsection{Target Problems}
Based on the above definitions, this work considers the following analytical and optimization problems:
\begin{itemize}
\item \textbf{Analysis on outage probability}: deriving a tractable analytical characterization of $P_{\mathrm{out}}$ and its bounds under the considered FAS-ARIS system model.
\item \textbf{Outage minimization}: minimizing $P_{\mathrm{out}}$ with respect to the ARIS configuration under practical power constraints.
\item \textbf{Throughput maximization}: maximizing $T(R)$.
\end{itemize}
These are analyzed and solved in Sections~\ref{sec3}-\ref{thput}.}
\section{Analysis on Outage Probability}
\label{sec3}
In this section, we provide a detailed analysis of the outage probability of the considered FAS-ARIS systems. {Note that although several mathematical tools employed in this paper are individually well established, the technical contribution of this analysis lies in integrating these tools into a unified outage analysis framework tailored to the FAS-ARIS architecture. By explicitly accounting for block-wise FAS correlation, ARIS-induced amplification noise, and limited CSI, the proposed formulation enables an analytically exact and tractable characterization of the outage probability, which has not been reported in existing FAS-ARIS studies~\cite{FASARIS22, FASARIS, fasaristnse}. More importantly, this analytical framework provides the foundation for the optimization framework developed in the subsequent sections, where the derived outage expressions are exploited to systematically design the ARIS configuration and optimize the throughput performance. Such a principled design would be difficult to achieve through simulation-only evaluations.}
\subsection{Optimal Amplification Gain $\rho^*$ of FAS-ARIS Systems}
\label{oag}
According to~(\ref{aris}), the SNR $\gamma_k$ at the $k$-th port of MU can be expressed as $\gamma_k = \frac{PA_k^2}{B_k^2 \sigma^2 +\frac{\sigma_0^2}{\rho^2}}$, where $A_k \triangleq |\mathbf{h}_k^* \mathbf{\Phi} \mathbf{g}|$ is the effective channel gain, and $B_k \triangleq || \mathbf{h}_k^* \mathbf{\Phi} ||_2$ quantifies the dynamic noise amplification due to ARIS. Hence, the SNR $\gamma$ of MU is given by~\cite{FAS}
\begin{equation}
\label{snrmax}
\gamma\triangleq\max_{k\in\mathbb{N}} \gamma_k= \max_{k\in\mathbb{N}} \frac{PA_k^2}{B_k^2 \sigma^2 +\frac{\sigma_0^2}{\rho^2}}.\footnote{{If the BS transmits multiple user streams simultaneously, the received signal at user $k\in\{1,\cdots,K\}$ contains inter-user interference, and the resulting SINR at the $n$th FAS port is given by
\begin{equation}
\label{sinrkn_fn}
\begin{cases}
\mathrm{SINR}_{k,n}=\frac{P\rho^2\big|\mathbf h_{k,n}^\ast \boldsymbol{\Phi}\mathbf G \mathbf f_k\big|^2}
{P\rho^2\sum_{i\neq k}\big|\mathbf h_{k,n}^\ast \boldsymbol{\Phi}\mathbf G \mathbf f_i\big|^2
+\rho^2\sigma^2\|\mathbf h_{k,n}^\ast\boldsymbol{\Phi}\|^2+\sigma_0^2},\\
\gamma_k=\max_{n}\mathrm{SINR}_{k,n}.
\end{cases}
\end{equation}
A full multi-user QoS optimization would then require the joint design of $\{\mathbf f_k\}$, user scheduling, $\rho$, and $\boldsymbol{\Phi}$, leading to a highly coupled and non-convex problem. The proposed single-user FAS-ARIS system should therefore be viewed as a desired-signal-boosting building block, while a full multi-user extension of the FAS-ARIS system would require an alternating-optimization (AO) framework that jointly balances desired-signal enhancement and interference suppression, as preliminarily explored in~\cite{FASARIS}.}}
\end{equation}
From~(\ref{snrmax}), since every $\{\gamma_k\}$ increases with $\rho$, so does $\gamma$, and it is clear that the optimal $\rho^*$ that maximizes every $\{\gamma_k\}$ is, by~(\ref{refl}), given by
\begin{equation}
\label{optrho}
\rho^* = \min\left(\sqrt {\frac{P'}{M(P\beta +\sigma^2)}}, \rho_{\max}\right),
\end{equation}
where $\rho_{\max}$ is the maximum amplification gain of the ARIS.
\begin{remark}
\label{rhodet}
{Since the ARIS is assumed to have access only to statistical CSI, adapting element-wise heterogeneous amplification gains $\{\rho_m\}_{m=1}^M$ to instantaneous channel realizations is hard to implement. In general, optimizing $\{\rho_m\}_{m=1}^M$ under fading channels would imply an instantaneous gain-control policy $\{\rho_m(\{\mathbf h_k\})\}_{m=1}^M$, as the resulting $\gamma$ is itself a random quantity determined by the random variables $\{A_k\}_{k=1}^N$ and $\{B_k\}_{k=1}^N$. Such a policy would require real-time CSI acquisition and frequent gain updates at the surface, leading to prohibitive overhead and contradicting the statistical-CSI assumption~\cite{FASRIS, FASARIS22}. Therefore, we adopt a deterministic long-term ARIS amplification gain $\rho$ and set $\forall\rho_m=\rho$, which preserves analytical tractability and yields a robust and implementable FAS-ARIS configuration design aligned with practical ARIS hardware constraints.}
\end{remark}
\begin{remark}
\label{revrm1}
{In practical deployment scenarios, even when the dominant LoS path is blocked, weak direct components due to scattering and diffraction may still exist, resulting in a non-zero but weak direct BS-MU link. This is because the direct BS-MU path experiences significantly larger path-loss than the short-range LoS BS-ARIS link, while the ARIS-assisted signal benefits from coherent beamforming and amplification gain, thereby rendering the residual direct component comparatively negligible. When a weak direct BS-MU link $\{h_{d,k}\}$ is coherently combined with the ARIS-assisted signal, the resulting received SNR at port $k$ is given by
\begin{equation}
\label{gkcoh}
\gamma_k^{\mathrm{(coh)}}=\frac{P\big|\rho\mathbf h_k^{*}\Phi\mathbf g + h_{d,k}\big|^2}
{\rho^2\sigma^2 B_k^2 + \sigma_0^2},
\end{equation}
By the triangle inequality, the coherent combination admits the following sandwich bound:
\begin{equation}
\label{sand11}
\frac{P\big(\rho A_k - |h_{d,k}|\big)^2}{\rho^2\sigma^2 B_k^2 + \sigma_0^2}
\le
\gamma_k^{\mathrm{(coh)}}
\le
\frac{P\big(\rho A_k + |h_{d,k}|\big)^2}{\rho^2\sigma^2 B_k^2 + \sigma_0^2}.
\end{equation}
Herein, if the direct link is weak in the sense that $|h_{d,k}|\le \epsilon\rho A_k~(\epsilon\ll1)$, the above inequality reduces to
\begin{equation}
\label{sand22}
(1-\epsilon)^2\gamma_k \le \gamma_k^{\mathrm{(coh)}}\le(1+\epsilon)^2\gamma_k.
\end{equation}
This result shows that, as long as the direct BS-MU link remains sufficiently weak, the effect of coherent superposition between $\rho\mathbf h_k^{*}\Phi\mathbf g$ and $h_{d,k}$, including both constructive and destructive phase alignments, is confined within a bounded multiplicative factor $(1\pm \epsilon)^2$ in~\eqref{sand22}. As a consequence, the performance trends are preserved up to this bounded perturbation.}
\end{remark}
\begin{remark}
\label{multinb}
{To assess whether the performance gains of the proposed FAS-ARIS architecture can be replicated by simply increasing the number of BS antennas $N_b$, we examine the effect of $N_b$ under the same signal model, ARIS amplification constraint, and FAS port selection rule. By extending the BS-ARIS channel to $\mathbf G\in\mathbb C^{M\times N_b}$ and applying maximum-ratio-combining (MRT) at BS, $A_k^{(N_b)2}$ corresponds to $A_k^2$ scales as
\begin{equation}
\label{hkgks}
A_k^{(N_b)2} =\big|\mathbf h_k^{*}\boldsymbol\Phi \mathbf G \mathbf f^\star\big|^2=N_b\big|\mathbf h_k^{*}\boldsymbol\Phi {\mathbf g}\big|^2=N_b A_k^2.
\end{equation}
Meanwhile, enforcing the ARIS reflection power constraint yields an $N_b$-dependent optimal amplification gain
\begin{equation}
\label{optrhonb}
\rho^\star(N_b)=\min\left(
\sqrt{\frac{P'}{M(P\beta N_b+\sigma^2)}}, 
\rho_{\max}
\right).
\end{equation}
Substituting $\rho^\star(N_b)$ into the SNR expression and dividing both the numerator and denominator by $\rho^{\star 2}(N_b)$, the resulting SNR at the $k$th FAS port can be written as
\begin{equation}
\label{gknbw}
\gamma_k(N_b)=\frac{PN_b A_k^2}
{B_k^2\sigma^2+\frac{\sigma_0^2}{\rho^{\star 2}(N_b)}}.
\end{equation}
Equation~\eqref{gknbw} shows that the $N_b$-antenna case retains the same functional structure as the proposed FAS-ARIS model, with the only structural difference being an $N_b$-fold scaling of the effective incident-link power, reflected to the denominator of~\eqref{gknbw}, and an accompanying adjustment of the effective noise level via $\rho^\star(N_b)$. As a result, the entire analytical framework remains applicable after appropriately dividing the target outage threshold by $N_b$ and replacing the relevant scalar parameters with their $N_b$-dependent counterparts.}
\end{remark}
\subsection{Formulation of Outage Probability ${P}_{\mathrm{out}}$}
By applying $\rho^*$, the outage probability ${P}_{\mathrm{out}}$ at MU for given target rate $R$ is given in~\eqref{outage}.
Herein, by~(\ref{snrmax}) the term inside $\mathbb{P}[\cdot]$ is equivalent with
\begin{equation}
\label{equi}
\begin{aligned}
\max_{k\in\{1, \cdots, N\}} \frac{PA_k^2}{B_k^2 \sigma^2 +\frac{\sigma_0^2}{\rho^{*2}}} <2^R -1
&\leftrightarrow A_k \le \sqrt {p_1 B_k^2 +p_2 }~(\forall k),
\end{aligned}
\end{equation}
where $p_1 = (2^R -1) \frac{\sigma^2 }{P},~p_2=(2^R -1)\frac{\sigma_0^2 }{P\rho^{*2} }$.
{Hence, by conditioning $B_k^2=s_k~(\forall k)$ and integrating the probability of~\eqref{equi} with respect to the PDF $f_{\{B_k^2\}_{k=1}^N} (\cdot)$ of $\{B_k\}_{k=1}^N$,} ${P}_{\mathrm{out}}$ can be written as
\begin{equation}
\label{outint}
\begin{aligned}
{P}_{\mathrm{out}}=\int_0^\infty \cdots \int_0^\infty &F_{\{A_k\}_{k=1}^N} \left(\{ \sqrt {p_1 s_k +p_2 } \}_{k=1}^N\right)\\ &f_{\{B_k^2\}_{k=1}^N} \left(\{s_k\}_{k=1}^N\right) dS,
\end{aligned}
\end{equation}
where $dS\triangleq ds_1 \cdots ds_N$ and $F_{\{A_k\}_{k=1}^N} (\cdot)$ is the joint CDF of $\{A_k\}_{k=1}^N$. In the following, we evaluate $F_{\{A_k\}_{k=1}^N}$ and $f_{\{B_k^2\}_{k=1}^N}$ by applying~(\ref{parcor}) to $A_k=|\mathbf{h}_k^*\mathbf{\Phi} \mathbf{g}|$ and $B_k=||\mathbf{h}_k^* \mathbf{\Phi} ||_2$.

\subsection{CDF of $\{A_k\}_{k=1}^N: F_{\{A_k\}_{k=1}^N}$}
The definition of $\{A_k\}_{k=1}^N$ is exactly same with the $\{A_k\}_{k=1}^N$ in~\cite{FASRIS}. Hence, by substituting $\{\sqrt{p_1 s_k +p_2 }\}_{k=1}^N$ to the CDF of $\{A_k\}_{k=1}^N$~\cite{FASRIS}, we can get
\begin{equation}
\label{CDFAk}
\begin{split}
&F_{\{A_k\}_{k=1}^N}(\{\sqrt{p_1 s_k +p_2 }\}_{k=1}^N)\\
=&\prod_{b=1}^B \int_0^\infty \frac{2r_b }{{\bar\sigma}^2\mu^2}e^{-\frac{r_b^2+|\eta|^2}{{\bar\sigma}^2\mu^2}} I_0 \left(\frac{2r_b |\eta|}{{\bar\sigma}^2 \mu^2 }\right)\\
&~~~~~~~~\prod_{k\in\mathbb{K}_b} \Biggl[1-Q_1 \Biggl(\sqrt{\frac{2}{{\bar\sigma}^2(1-\mu^2)}}r_b, \sqrt{\frac{2(p_1s_k +p_2)}{{\bar\sigma}^2(1-\mu^2)}}\Biggr)\Biggr]dr_b,
\end{split}
\end{equation}
where {$r_b$ denotes a realization of the random variable $\Lambda_b\triangleq |\eta+\mu \tilde {\mathbf h}_b \boldsymbol\Phi \mathbf g^*|$~\cite{FASRIS}}, ${\bar\sigma}^2 \triangleq \frac{M\alpha\beta}{K+1}, \eta\triangleq\sqrt{\frac{\alpha K}{K+1}}\bar{\mathbf{h}}^*\mathbf{\Phi} \mathbf{g}$, $\mathbb{K}_b$ is the index set which is a partition of $\mathbb{N}$ $(\cup_{b=1}^B \mathbb{K}_b=\mathbb{N}, \mathbb{K}_b \cap \mathbb{K}_{b'}=\emptyset~(b\neq b'))$ and corresponds to $\mathbf{C}_b$, $Q_1 (\cdot, \cdot)$ is the first-order Marcum $Q$-function, and $I_\nu (x)$ is a $\nu$-th-order modified Bessel function of the first-kind. 
\subsection{PDF of $\{B_k^2\}_{k=1}^N: f_{\{B_k^2\}_{k=1}^N}$}
For $k\in\mathbb{K}_b$, since $B_k^2 =||\mathbf{h}_k^* \mathbf{\Phi}||_2^2=||\mathbf{h}_k||_2^2$, 
$B_k^2$ has the same distribution with
\begin{equation}
\label{bksd}
||\mathbf{h}_k||_2^2=\left|\left| \tilde{\pmb{\eta}}' + \mu \left[ \tilde{h}_b^{(m)} \right]_{m=1}^M + \sqrt{1 - \mu^2} \left[ e_k^{(m)} \right]_{m=1}^M \right|\right|_2^2,
\end{equation}
where {$\tilde{\pmb{\eta}}' = \sqrt{\frac{\alpha K}{K+1}} \bar{\mathbf{h}}$} and the scalars in~(\ref{bksd}) represent the individual $m$-th elements of the respective matrix and vector variables. By denoting
\begin{equation}
\label{dnt}
\tilde{\pmb{\eta}}' + \mu \left[ \tilde{h}_b^{(m)} \right]_{m=1}^M= \bar{\pmb{\Lambda}}_b',  \left|\left| \bar{\pmb{\Lambda}}_b' \right|\right|_2 = \pmb{\Lambda}_b',
\end{equation}
Eq.~(\ref{bksd}) becomes
\begin{equation}
\label{chiex}
\begin{aligned}
&\left|\left| \bar{\pmb{\Lambda}}_b' + \sqrt{1 - \mu^2} \left[ e_k^{(m)} \right]_{m=1}^M \right|\right|_2^2\\
&=\frac{\alpha(1-\mu^2)}{2(K+1)} \left|\left|\underbrace{\frac{\sqrt{2(K+1)}}{\sqrt{\alpha(1 - \mu^2)}}  \bar{\pmb{\Lambda}}_b' 
+ \sqrt{\frac{2(K + 1)}{\alpha}} \left[ e_k^{(m)} \right]_{m=1}^M}_{[\bar{e}_k^{(m)}]_{m=1}^M} \right|\right|_2^2.
\end{aligned}
\end{equation}
From~(\ref{chiex}), it is clear that
\begin{equation}
\label{bare}
\bar{e}_k^{(m)} \sim \mathcal{CN}\left(\frac{\sqrt{2(K+1)}}{\sqrt{\alpha(1 - \mu^2)}} \bar{\pmb{\Lambda}}_b^{'(m)}, 2 \right)~(m=1, \cdots, M),
\end{equation}
which implies that the real and complex component of $\bar{e}_k^{(m)}$ has unit variance, and $\{\bar{e}_k^{(m)}\}_{m=1}^M$ are all independent because of the definition of $\mathbf{e}_k$. Therefore, $B_k^2$ conditioned by $\pmb{\Lambda}_b'=s_b$ follows a non-central chi-squared distribution scaled by~$C\triangleq\frac{\alpha(1-\mu^2)}{2(K+1)}$, with DoF $2M$ and the non-centrality parameter $\lambda= \frac{1}{C} \sum_{m=1}^M \left|\bar{\pmb{\Lambda}}_b^{'(m)}\right|^2 =\frac{ s_b^2} {C}$. The conditional PDF is given by~\cite{Lap}
\begin{equation}
\label{Bklk}
\begin{aligned}
f_{B_k^2 |\pmb{\Lambda}_b' =s_b} (s_k  )=&\frac{K+1}{\alpha(1-\mu^2)} e^{-\frac{s_k + s_b^2}{\frac{\alpha(1-\mu^2)}{K+1}}} \left(\frac{s_k}{s_b^2}\right)^{\frac{M-1}{2}}\\
& I_{M-1} \left( \frac{2s_b (K+1)\sqrt {s_k}}{\alpha(1-\mu^2)} \right).
\end{aligned}
\end{equation}

Conditioned on $\pmb{\Lambda}_b'$, it is clear that $\{B_k^2 \}_{k\in\mathbb{K}_b}|{\pmb{\Lambda}_b'=s_b}$ are mutually independent. This then allows us to deduce the joint PDF of $\{B_k^2 \}_{k\in\mathbb{K}_b}|{\pmb{\Lambda}_b'=s_b}$ as
\begin{equation}
\label{Bklk22}
\begin{aligned}
&f_{\{B_k^2 \}_{k\in\mathbb{K}_b}|\pmb{\Lambda}_b' =s_b} \left(\{s_k\}_{k\in\mathbb{K}_b}\right)\\
&=\prod_{k\in\mathbb{K}_b}\frac{K+1}{\alpha(1-\mu^2)} e^{-\frac{s_k + s_b^2}{\frac{\alpha(1-\mu^2)}{K+1}}} \left(\frac{s_k}{s_b^2}\right)^{\frac{M-1}{2}}\\
&~~~~~~~~~~~I_{M-1} \left( \frac{2s_b (K+1)\sqrt {s_k}}{\alpha(1-\mu^2)} \right).
\end{aligned}
\end{equation}
To obtain the unconditioned joint PDF,~(\ref{Bklk22}) should be averaged over the distribution of $\pmb{\Lambda}_b'=\left|\left|\tilde{\pmb{\eta}}' + \mu \left[ \tilde{h}_b^{(m)} \right]_{m=1}^M\right|\right|_2$. 
By manipulating it,
\begin{equation}
\label{ncchi}
\begin{aligned}
&\left|\left|\tilde{\pmb{\eta}}' + \mu \left[ \tilde{h}_b^{(m)} \right]_{m=1}^M\right|\right|_2\\
&= \mu \sqrt{\frac{\alpha}{2(K+1)}} 
\left| \left|\underbrace{\frac{1}{\mu} \sqrt{\frac{2(K+1)}{\alpha}} \tilde{\pmb{\eta}}' 
+  \sqrt{\frac{2(K+1)}{\alpha}} \left[ \tilde{h}_b^{(m)} \right]_{m=1}^M}_{[\bar{\tilde{h}}_b^{(m)}]_{m=1}^M} \right|\right|_2,
\end{aligned}
\end{equation}
and it is clear that
\begin{equation}
\label{cn22}
\bar{\tilde{h}}_b^{(m)}\sim \mathcal{CN} \left( \frac{1}{\mu} \sqrt{\frac{2(K+1)}{\alpha}} \tilde{\pmb{\eta}}^{'(m)},  2 \right)~(m=1, \cdots, M),
\end{equation}
which implies that the real and complex component of $\bar{\tilde{h}}_b^{(m)}$ has unit variance, and $\{\bar{\tilde{h}}_b^{(m)}\}_{m=1}^M$ are all independent because of the definition of $\tilde{\mathbf{h}}_b$. Hence, we can conclude that $\pmb{\Lambda}_b'$ follows a non-central chi distribution scaled by $C'\triangleq \frac{\mu}{\sqrt{2}}\sqrt{\frac{\alpha}{K+1}}$, DoF $2M$ and non-centrality parameter $\lambda'= \frac{1}{C'} \sqrt{\sum_{m=1}^M \left|\tilde{\pmb{\eta}}^{'(m)}\right|^2} =\frac{||\tilde{\pmb{\eta}}'||_2} {C'}=
\frac{\sqrt{2 KM}}{\mu}$. The PDF is given by~\cite{novfad}
\begin{equation}
\label{ncchi}
\begin{aligned}
f_{\Lambda_b'}(s_b) =& \frac{\sqrt{2(K+1)}}{\mu \sqrt{\alpha}}  \frac{e^{-\frac{s_b^2 + \frac{\alpha KM}{K+1}}{\frac{\mu^2 \alpha}{K+1}}} \left( \frac{s_b}{\frac{\mu}{\sqrt{2}}\sqrt{\frac{\alpha}{K+1}}} \right)^{2M} \frac{\sqrt{2 KM}}{\mu}}{\left( \frac{\sqrt{4 K(K+1)M}}{\mu^2 \sqrt{\alpha}} s_b \right)^M}\\
&I_{M-1} \left(  \frac{\sqrt{4 K(K+1)M}}{\mu^2 \sqrt{\alpha}} s_b \right).
\end{aligned}
\end{equation}
Hence, the PDF of $\{B_k^2\}_{k\in\mathbb{K}_b}$ is given by
\begin{equation}
\label{pdfbkb}
\begin{aligned}
&f_{\{B_k^2 \}_{k\in\mathbb{K}_b}} \left(\left\{s_k\right\}_{k\in\mathbb{K}_b}\right)\\
&=\int_0^\infty \frac{\sqrt{2(K+1)}}{\mu \sqrt{\alpha}}  \frac{e^{-\frac{s_b^2 + \frac{\alpha KM}{K+1}}{\frac{\mu^2 \alpha}{K+1}}} \left( \frac{s_b}{\frac{\mu}{\sqrt{2}}\sqrt{\frac{\alpha}{K+1}}} \right)^{2M} \frac{\sqrt{2 KM}}{\mu}}{\left( \frac{\sqrt{4 K(K+1)M}}{\mu^2 \sqrt{\alpha}} s_b \right)^M}\\
&~~~~~~~~~I_{M-1} \left(  \frac{\sqrt{4 K(K+1)M}}{\mu^2 \sqrt{\alpha}} s_b \right) \\
&~~~~~~~~~\prod_{k\in\mathbb{K}_b}\frac{K+1}{\alpha(1-\mu^2)} e^{-\frac{s_k + s_b^2}{\frac{\alpha(1-\mu^2)}{K+1}}} \left(\frac{s_k}{s_b^2}\right)^{\frac{M-1}{2}}\\
&~~~~~~~~~~~~~~~~I_{M-1} \left( \frac{2s_b (K+1)\sqrt {s_k}}{\alpha(1-\mu^2)}\right)ds_b,
\end{aligned}
\end{equation}
and since the ports (indices) from different block has zero-correlation by assumption, the PDF of $\{B_k^2\}_{k=1}^N$ is given by
\begin{equation}
\label{pdfbktotal}
\begin{aligned}
&f_{\{B_k^2 \}_{k=1}^N} \left(\left\{s_k\right\}_{k=1}^N\right)\\
&=\prod_{b=1}^B\int_0^\infty \frac{\sqrt{2(K+1)}}{\mu \sqrt{\alpha}}  \frac{e^{-\frac{s_b^2 + \frac{\alpha KM}{K+1}}{\frac{\mu^2 \alpha}{K+1}}} \left( \frac{s_b}{\frac{\mu}{\sqrt{2}}\sqrt{\frac{\alpha}{K+1}}} \right)^{2M} \frac{\sqrt{2 KM}}{\mu}}{\left( \frac{\sqrt{4 K(K+1)M}}{\mu^2 \sqrt{\alpha}} s_b \right)^M}\\
&~~~~~~~~~~~~~I_{M-1} \left(  \frac{\sqrt{4 K(K+1)M}}{\mu^2 \sqrt{\alpha}} s_b \right) \\
&~~~~~~~~~~~~~\prod_{k\in\mathbb{K}_b}\frac{K+1}{\alpha(1-\mu^2)} e^{-\frac{s_k + s_b^2}{\frac{\alpha(1-\mu^2)}{K+1}}} \left(\frac{s_k}{s_b^2}\right)^{\frac{M-1}{2}}\\
&~~~~~~~~~~~~~~~~~~~~I_{M-1} \left( \frac{2s_b (K+1)\sqrt {s_k}}{\alpha(1-\mu^2)}\right)ds_b,
\end{aligned}
\end{equation}
\subsection{Obtaining ${P}_{\mathrm{out}}$}
By applying~(\ref{pdfbktotal}) and~(\ref{CDFAk}) to~(\ref{outint}), ${P}_{\mathrm{out}}$ is given by~(\ref{pout}). 
\begin{figure*}
\begin{equation}
\label{pout}
\begin{aligned}
{P}_{\mathrm{out}}=&\int_0^{\infty}\cdots\int_0^{\infty}\prod_{b=1}^B \int_0^\infty \int_0^\infty\frac{2r_b }{{\bar\sigma}^2\mu^2}e^{-\frac{r_b^2+|\eta|^2}{{\bar\sigma}^2\mu^2}} I_0 \left(\frac{2r_b |\eta|}{{\bar\sigma}^2 \mu^2 }\right)\frac{\sqrt{2(K+1)}}{\mu \sqrt{\alpha}}  \frac{e^{-\frac{s_b^2 + \frac{\alpha KM}{K+1}}{\frac{\mu^2 \alpha}{K+1}}} \left( \frac{s_b}{\frac{\mu}{\sqrt{2}}\sqrt{\frac{\alpha}{K+1}}} \right)^{2M} \frac{\sqrt{2 KM}}{\mu}}{\left( \frac{\sqrt{4 K(K+1)M}}{\mu^2 \sqrt{\alpha}} s_b \right)^M}\\
& I_{M-1} \left(  \frac{\sqrt{4 K(K+1)M}}{\mu^2 \sqrt{\alpha}} s_b \right)\prod_{k\in\mathbb{K}_b} \Biggl(1-Q_1 \Biggl(\sqrt{\frac{2}{{\bar\sigma}^2(1-\mu^2)}}r_b, \sqrt{\frac{2}{{\bar\sigma}^2(1-\mu^2)}}\sqrt{p_1s_k +p_2 }\Biggr)\Biggr) \\
&~~~~~~~~~~~~~~~~~~~~~~~~~~~~~~~~~~~~~~~~ \frac{K+1}{\alpha(1-\mu^2)} e^{-\frac{s_k + s_b^2}{\frac{\alpha(1-\mu^2)}{K+1}}} \left(\frac{s_k}{s_b^2}\right)^{\frac{M-1}{2}} I_{M-1} \left( \frac{2s_b (K+1)\sqrt {s_k}}{\alpha(1-\mu^2)}\right)dr_bds_bdS.
\end{aligned}
\end{equation}
\hrule
\end{figure*}
\section{Upper-Bound of~${P}_{\mathrm{out}}$}
As shown in~\cite{newcor}, when $\{L_b\}_{b=1}^B$ is fixed, ${P}_{\mathrm{out}}$ calculated based on the independent antennas equivalent (IAE) model by letting $\mu=1$, which is denoted by $\hat{{P}}_{\mathrm{out}}$, provides an upper-bound of ${P}_{\mathrm{out}}$ derived using the BDMA model in~(\ref{pout}). In detail, when $\mu=1$, it is clear that~(\ref{parcor}) becomes
\begin{equation}
\label{outint22}
\begin{aligned}
\hat{{P}}_{\mathrm{out}}=\int_0^\infty \cdots \int_0^\infty F_{\{A_b\}_{b=1}^B} \left(\{ \sqrt {p_1 s_b +p_2 } \}_{b=1}^B\right)\\ f_{\{B_b^2\}_{b=1}^B} \left(\{s_b\}_{b=1}^B\right) ds_1 \cdots ds_B.
\end{aligned}
\end{equation}
Here, the joint CDF of $\{A_b\}_{b=1}^B$ with substitution of $\{\sqrt{p_1 s_b +p_2 }\}_{b=1}^B$ is given by~\cite{FASRIS}
\begin{equation}
\label{abcdf}
\begin{aligned}
&F_{\{A_b\}_{b=1}^B} \left(\{\sqrt{p_1 s_b +p_2 }\}_{b=1}^B\right)\\
&=\prod_{b=1}^{B} \left[1 - Q_1\left( \sqrt{\frac{2}{\bar{\sigma}^2}} |\eta|, \sqrt{\frac{2(p_1 s_b +p_2)}{\bar{\sigma}^2}} \right) \right],
\end{aligned}
\end{equation}
which is due to the fact that
\begin{equation}
\label{abdist}
A_b=|\mathbf{h}_b^* \mathbf{\Phi} \mathbf{g}| = \left| \sqrt{\frac{\alpha K}{K+1}} \bar{\mathbf{h}}^*\mathbf{\Phi} \mathbf{g} + \tilde{\mathbf{h}}_b^*\mathbf{\Phi} \mathbf{g}\right|
\end{equation}
follows the Rician distribution by definition. For $f_{\{B_b^2\}_{b=1}^B} \left(\{s_b\}_{b=1}^B\right)$, it is clear that
\begin{equation}
\label{defbb}
\begin{aligned}
B_b ^2 &=\left| \left|\tilde{\pmb{\eta}}' + \mu \left[ \tilde{h}_b^{(m)} \right]_{m=1}^M\left(= \bar{\pmb{\Lambda}}_b'\right) \right|\right|_2^2\\
&=C^{'2}\left|\left| [\bar{\tilde{h}}_b^{(m)}]_{m=1}^M \right|\right|_2^2~(\mu =1),
\end{aligned}
\end{equation}
and by~(\ref{cn22}), 
we can conclude that $\{B_b^2\}_{b=1}^B$ are independent and each $B_b^2$ follow a non-central chi-squared distribution scaled by $C^{'2}\left(= \frac{\alpha}{2(K+1)}\right)$, DoF $2M$ and non-centrality parameter $ \frac{1}{C^{'2}} {\sum_{m=1}^M \left|\tilde{\pmb{\eta}}^{'(m)}\right|^2} =\frac{|\tilde{\pmb{\eta}}'|^2} {C^{'2}}=\lambda^{'2}=2 KM$. Hence, $f_{\{B_b^2\}_{b=1}^B} \left(\{s_b\}_{b=1}^B\right)$ can be evaluated by
 \begin{equation}
 \label{pdfbb}
 \begin{aligned}
f_{\{B_b^2\}_{b=1}^B} \left(\{s_b\}_{b=1}^B\right)=\prod_{b=1}^B& \frac{K+1}{\alpha} e^{-\frac{s_b +\frac{\alpha KM}{K+1}}{\frac{\alpha}{K+1}}} \left( \frac{s_b}{\frac{\alpha KM}{K+1}} \right)^{\frac{M-1}{2}}\\
&I_{M-1} \left( \sqrt{ \frac{4K(K+1)M s_b}{\alpha} } \right).
\end{aligned}
\end{equation}
Therefore, by applying~(\ref{abcdf}) and~(\ref{pdfbb}) to~(\ref{outint22}), the upper-bound $\hat{{P}}_{\mathrm{out}}$ is given by~(\ref{poutupp}). 
{It is worth emphasizing that, although $\hat{{P}}_{\mathrm{out}}$ in~\eqref{poutupp} is derived under the IAE assumption with $\mu=1$, its role in this work is fundamentally different from existing FAS-RIS studies such as~\cite{FASRIS}. In particular, the considered FAS-ARIS system involves a substantially more intricate statistical structure due to the presence of ARIS amplification noise and its coupling with spatially correlated FAS ports, which makes the outage analysis significantly more challenging than in conventional FAS-RIS settings. Within this more complex signal and noise model, $\hat{{P}}_{\mathrm{out}}$ in~\eqref{poutupp} is not introduced as a standalone inequality, but as a key component of the proposed system-level analytical framework that enables a tractable characterization of the outage probability. Furthermore, it serves as an analytical benchmark that directly enables the derivation of the near-optimal rate solution and the subsequent throughput optimization in Section~\ref{thput} without exhaustive search. This benchmark-driven analytical design framework is specific to the considered FAS-ARIS architecture and does not arise in previous analyses.} Furthermore, in Section~\ref{numr}, we will numerically demonstrate that $\hat{{P}}_{\mathrm{out}}$ closely approximates the actual ${P}_{\mathrm{out}}$, validating the tightness and practical relevance of the IAE-based upper-bound.
\begin{figure*}
\begin{equation}
\label{poutupp}
\begin{aligned}
\hat{{P}}_{\mathrm{out}}=&\int_0^\infty \cdots \int_0^\infty \prod_{b=1}^{B} \left[1 - Q_1\left( \sqrt{\frac{2}{\bar{\sigma}^2}} |\eta|, \sqrt{\frac{2(p_1 s_b +p_2)}{\bar{\sigma}^2}} \right) \right]\frac{K+1}{\alpha} e^{-\frac{s_b +\frac{\alpha KM}{K+1}}{\frac{\alpha}{K+1}}} \left( \frac{s_b}{\frac{\alpha KM}{K+1}} \right)^{\frac{M-1}{2}} I_{M-1} \left( \sqrt{  \frac{4K(K+1)M s_b}{\alpha}} \right) dS\\
=&\left(\int_0^\infty \left[1 - Q_1\left( \sqrt{\frac{2}{\bar{\sigma}^2}} |\eta|, \sqrt{\frac{2(p_1 s +p_2)}{\bar{\sigma}^2}} \right) \right]\underbrace{\frac{K+1}{\alpha} e^{-\frac{s +\frac{\alpha KM}{K+1}}{\frac{\alpha}{K+1}}} \left( \frac{s}{\frac{\alpha KM}{K+1}} \right)^{\frac{M-1}{2}} I_{M-1} \left( \sqrt{  \frac{4K(K+1)M s}{\alpha}} \right)}_{\triangleq \bar{g}(s)}ds \right)^B
\end{aligned}
\end{equation}
\hrule
\end{figure*}
\section{Outage Probability Optimization}
\label{opop}
In this section, we aim to minimize $P_{\mathrm{out}}$ by optimizing $\mathbf{\Phi}$, formulated as follows:
\begin{equation}
\label{outopt}
\min_{\mathbf{\Phi}}~{P}_{\mathrm{out}}~\mathrm{s.t.}~\theta_m\in[0,2\pi)~(\forall m).
\end{equation}
We note that $|\eta|$ is the only term involving $\mathbf{\Phi}$, which implies that minimizing ${P}_{\mathrm{out}}$ with respect to $\mathbf{\Phi}$ reduces to minimizing it with respect to $|\eta|$. However, since the integrand of ${P}_{\mathrm{out}}$ in~(\ref{pout}) includes the following term:
\begin{equation}
\label{etaout}
u(r_b, t) \triangleq e^{-\frac{Bt^2}{{\bar\sigma}^2\mu^2}} \prod_{b=1}^B I_0\left(\frac{2r_b t}{{\bar\sigma}^2 \mu^2}\right)~(|\eta|=t),
\end{equation}
within a total of $N+2B$ integrals, it is analytically intractable to optimize ${P}_{\mathrm{out}}$ directly with respect to $t$, especially due to $I_0$ which intricates the integral variable $r_b$ and $t$ inside. {To gain insight into its behavior, we consider the monotonicity of ${P}_{\mathrm{out}}$ by leveraging first-order stochastic dominance as a function of $t$. It provides a comprehensive and rigorous characterization of the dependence of ${P}_{\mathrm{out}}$ on $t$, enabling analytical insight despite the intractability of the exact expression.}

{
\begin{theorem}
\label{thmnd}
$P_{\mathrm{out}}(t)$ is monotonically non-increasing in $t$.
\end{theorem}
\begin{proof}
\textbf{Step 1 (Conditional monotonicity in $\Lambda_b$):} Based on the conditional formulation in Section~\ref{sec3} and the definition of $\Lambda_b$, the outage probability can be expressed as
\begin{equation}
\label{pot}
P_{\mathrm{out}}(t)=\mathbb{E}_{\{\Lambda_b(t)\}}
\left[
\prod_{b=1}^{B} \phi_b(\Lambda_b(t))
\right],
\end{equation}
where the conditional outage term $\phi_b(\Lambda_b)$ is given by
\begin{equation}
\label{defphibb}
\begin{aligned}
\phi_b(\Lambda_b)
&\triangleq\mathbb{P}\left[A_k \le \sqrt{\gamma_{\mathrm{th}}}, \forall k\in\mathcal{K}_b \big| \Lambda_b\right]\\
&=\left[1-Q_1\left(
\sqrt{\frac{2}{\bar{\sigma}^2(1-\mu^2)}}\Lambda_b,
\sqrt{\frac{2\gamma_{\mathrm{th}}}{\bar{\sigma}^2(1-\mu^2)}}
\right)\right]^{L_b},
\end{aligned}
\end{equation}
since each $\Lambda_b(t)$ follows a Rician distribution whose noncentrality parameter is $t$~\cite{FASRIS}. Thereafter, because $1-Q_1(a,b)$ is monotonically decreasing in $a$, $\phi_b(\Lambda_b)$ is monotonically decreasing in $\Lambda_b$. Consequently, the product
\begin{equation}
\label{phiprod}
\Phi(\boldsymbol{\lambda})
\triangleq
\prod_{b=1}^{B} \phi_b(\lambda_b)
\end{equation}
is componentwise decreasing in $\boldsymbol{\lambda}=[\lambda_1\cdots \lambda_B]^{\mathrm T}$.\\
\textbf{Step 2 (Stochastic ordering of $\Lambda_b(t)$ in $t$):} We recall the definition of first-order stochastic dominance~\cite{sd}. For two random variables $X$ and $Y$, $X \succeq_{\mathrm{FOSD}} Y$ if
\begin{equation}
\label{fxfy}
F_X(x) \le F_Y(x)~(\forall x),
\end{equation}
or equivalently,
\begin{equation}
\label{egey}
\mathbb{E}[g(X)] \ge \mathbb{E}[g(Y)]~(\forall~\text{increasing functions}~g).
\end{equation}
The CDF of $\Lambda_b(t)$ is given by
\begin{equation}
\label{lcdf}
F_{\Lambda_b}(x;t)=1-Q_1\left(\frac{t}{\sigma},\frac{x}{\sigma}\right),
\end{equation}
which implies that for any $t_2 \ge t_1$,
\begin{equation}
\label{t2t1}
\begin{aligned}
&Q_1\left(\frac{t_2}{\sigma},\frac{x}{\sigma}\right)\ge Q_1\left(\frac{t_1}{\sigma},\frac{x}{\sigma}\right)\leftrightarrow F_{\Lambda_b}(x;t_2) \le F_{\Lambda_b}(x;t_1)~(\forall x).
\end{aligned}
\end{equation}
Hence, $\Lambda_b(t_2) \succeq_{\mathrm{FOSD}} \Lambda_b(t_1)$.\\
\textbf{Step 3 (Conclusion):}
Since $\Phi(\boldsymbol{\lambda})$ is decreasing in each component and $\{\Lambda_b(t)\}$ are independent across $b$, the stochastic dominance established in Step~2 implies
\begin{equation}
\label{t2t1s3}
t_2 \ge t_1~\rightarrow~\mathbb{E}\left[\Phi(\{\Lambda_b(t_2)\})\right]\le\mathbb{E}\left[\Phi(\{\Lambda_b(t_1)\})\right].
\end{equation}
Therefore, $P_{\mathrm{out}}(t)$ is monotonically non-increasing in $t$, which completes the proof.
\end{proof}
\begin{remark}
\label{monohat}
We can observe same effect in~(\ref{poutupp}), since the term $1 - Q_1\left( \sqrt{\frac{2}{\bar{\sigma}^2}} |\eta|, \sqrt{\frac{2(p_1 s_b +p_2)}{\bar{\sigma}^2}} \right)$ also diminishes when $|\eta|$ increases~\cite{marQ}, leads to the decrease of $\hat{{P}}_{\mathrm{out}}$.
\end{remark}
}

This property enables us to reformulate Problem~(\ref{outopt}) into an equivalent problem that can be solved more efficiently:
\begin{equation}
\label{outopt2}
\max_{\boldsymbol\Phi}~|\bar{\mathbf{h}}^*\mathbf{\Phi} \mathbf{g}|~\mathrm{s.t.}~\theta_m\in[0,2\pi)~(\forall m).
\end{equation}
It is clear that the optimal solution $\mathbf{\Phi}^\star\triangleq\mathrm{diag}\left(\{e^{j\theta_m^\star}\}_{m=1}^M\right)$ of Problem~(\ref{outopt2}) is obtained with
\begin{equation}
\label{optimalphi}
\theta_m^\star=\angle \bar{\mathbf{h}}^{(m)}-\angle {\mathbf{g}}^{(m)},~(m=1,\cdots,M).
\end{equation}
Furthermore, the following also holds under~(\ref{optimalphi}):
\begin{equation}
\label{aeta}
\begin{aligned}
\sqrt{\frac{2}{{\bar\sigma}^2}} |\eta|
&=\sqrt{\frac{2K}{M}} \sum_{m=1}^M |\bar{\mathbf{h}}^{*(m)}| |\bar{\mathbf{g}}^{(m)}|=\sqrt{2KM},
\end{aligned}
\end{equation}
which will be utilized in the throughput optimization process under the optimal phase alignment.
\begin{remark}
\label{phasenov}
{Although~\eqref{optimalphi} admit a simple phase-alignment form, they are not introduced heuristically nor borrowed from conventional RIS designs. Instead, they arise as the direct outcome of a rigorous outage probability minimization formulated for the considered FAS-ARIS system. In particular, from the raw outage expression in~\eqref{pout}, the optimal phase structure does not emerge trivially due to the coupled effects of FAS spatial correlation and ARIS amplification noise. The key analytical enabler is the monotonic relationship between $P_{\mathrm{out}}$ and $t$ established by Theorem~\ref{thmnd}, which allows the originally intractable outage minimization problem to be reformulated into a tractable maximization of $t$. This result formally connects outage probability minimization with the resulting phase-alignment structure and reveals that the simple alignment rule in~\eqref{optimalphi} is in fact the optimal solution of the underlying outage-oriented design problem for the FAS-ARIS architecture.}
\end{remark}
\section{Throughput Optimization}\label{thput}
{According to~(\ref{outage}), $P_{\mathrm{out}}(R)$ increases monotonically with respect to $R$. This monotonicity reveals a fundamental tradeoff between spectral efficiency and reliability: increasing $R$ enhances the nominal data rate but reduces the non-outage probability $1-P_{\mathrm{out}}(R)$, thereby degrading coverage performance. To explicitly capture this tradeoff, we have defined the throughput of the MU $T(R)$ in~\eqref{Tdef} as the effective throughput to be maximized, which measures the average successfully delivered information rate rather than the instantaneous capacity. 
Maximizing $T(R)$ therefore corresponds to selecting $R$ that best balances rate and reliability.
}

{However, $P_{\mathrm{out}}$ in~(\ref{pout}) involves multiple integral terms and is computationally intractable for direct optimization. To enable a practical and low-complexity throughput optimization, we adopt $\hat{P}_{\mathrm{out}}$ in~(\ref{poutupp}), which provides a tight and analytically tractable surrogate for $P_{\mathrm{out}}$~\cite{FASRIS, FASRIS22, newcor}. Accordingly, the throughput optimization problem is formulated as
\begin{equation}
\label{thpopt}
\max_{R}~ T(R) = R\big(1-\hat{P}_{\mathrm{out}}(R)\big)~ \mathrm{s.t.} ~ R \in [R_{\min}, R_{\max}].
\end{equation}}
It remains challenging to directly solve Problem~(\ref{thpopt}) due to the complex structure of $T$, which involves intricate components such as $1 - Q_1$ and $I_{M-1}$ within the integral. We, therefore, first substitute $R=\log_2 (1+x)$ such that $\check{T} \triangleq T(\log_2 (1 + x))$ to simplify the expression. Thereafter, we divide the domain into three subintervals 1) $x\le \Lambda_0$, 2) $x\ge\Lambda_1$, and 3) $x\in[\Lambda_0, \Lambda_1]$ based on the fact that different asymptotic approximations are valid in the low-, high-, and intermediate-regimes of $x$, enabling tractable and accurate characterization of the integral in each region. We then identify $x$ and corresponding $R$ that maximizes $\check{T}$ within each interval. Among these candidates, we then select the one that yields the global maximum of $T$.
\subsubsection{$x\le\Lambda_0$}
To determine $\Lambda_0$ first, we consider Theorem~\ref{quasi}.
\begin{theorem}
\label{quasi}
Let $R = \log_2 (1 + x)$ be substituted into~(\ref{Tdef}). Then, the resulting function $\check{T} \triangleq T(\log_2 (1 + x))$ is concave over the interval $x \in [0, \Lambda_0)$, where
\begin{equation}
\label{ld}
\Lambda_0 \triangleq \frac{1}{\bar{p}_1 U + \bar{p}_2} \left( \frac{1}{a} \psi^{-1}\left( \frac{1}{a^2} \right) \right)^2,
\end{equation}
$\psi^{-1}$ denotes the inverse of $\psi$~\cite{bessbook}, and $U > 0$ is a sufficiently large constant.
\end{theorem}
		\begin{proof}
See Appendix A.
\end{proof}
Based on Theorem~\ref{quasi}, to find $x^\star$ that maximizes $\check{T}$, we consider the following cases:
\begin{enumerate}
\item If $\frac{\partial \check{T}}{\partial x}|_{x=\Lambda_0} \ge 0$, then $x^\star \in [0, \Lambda_0]$ is $x^\star = \Lambda_0$. This is because $\check{T}$ is concave on $[0, \Lambda_0]$, and for any $x' \in [0, \Lambda_0]$, we have $\frac{\partial \check{T}}{\partial x}|_{x = x'} \ge \frac{\partial \check{T}}{\partial x}|_{x = \Lambda_0} \ge 0$, which implies that $\check{T}$ is monotonically increasing over the interval.
\item If $\frac{\partial \check{T}}{\partial x}|_{x = \Lambda_0} \le 0$, then $x^\star \in [0, \Lambda_0]$ is the solution to $\frac{\partial \check{T}}{\partial x} = 0$. Since $\check{T}$ is concave in this range, the solution is unique and can be efficiently obtained using the gradient ascent method~\cite{boyd}.
\end{enumerate}
The corresponding $R^\star$ that maximizes $\check{T}$ is then given by
\begin{equation}
\label{Rxstar}
R^\star = \log_2 (1+x^\star).
\end{equation} 
\begin{remark}
\label{r000}
Since $\frac{\partial \check{T}}{\partial x}=\frac{\partial \check{T}}{\partial R} \frac{\partial R}{\partial x}$ and $\frac{\partial R}{\partial x} = \frac{1}{\ln 2 (1+x)}$, we can check that $\frac{\partial \check{T}}{\partial x}$ and $\frac{\partial \check{T}}{\partial R}$ have same sign.
\end{remark}
\subsubsection{$x\ge\Lambda_1$}
For $x\ge \Lambda_0$, by approximation $(1-t)^B\approx 1-Bt$ for $t\ll1$, $\check{T}$ becomes
\begin{equation}
\label{happx}
\check{T} \approx B \log_2 (1+x) \int_0^\infty Q_1 \left(a, \sqrt{(\bar{p}_1 s+\bar{p}_2)x}\right)\bar{g}(s)ds,
\end{equation}
and to reduce the search space, we further narrow the range by determining $\Lambda_1$ according to following Theorem~\ref{max22}.
\begin{theorem}
	\label{max22}
	$\check{T}$ is either quasiconcave or monotonically decreasing function in $x\ge \Lambda_1$, where
	\begin{equation}
	\label{l0}
	\Lambda_1\triangleq\Lambda(0)=\frac{1}{\bar{p}_2} \left(\frac{1}{a}\psi^{-1}\left(\frac{1}{a^2}\right)\right)^2.
	\end{equation}
		\end{theorem}
		\begin{proof}
See Appendix B.
\end{proof}
Therefore, we can achieve maximum $\check{T}$ by unique $x^{\star\star}$ in $x\ge \Lambda_1$ by $x^{\star\star}=\max(\Lambda_1, x_\omega),$ which is equivalent to letting corresponding $R^{\star\star}$ by
\begin{equation}
\label{xr22}
R^{\star\star} = \log_2 (1+x^{\star\star}).
\end{equation}
\begin{remark}
\label{subre}
In Section~\ref{numr}, we will numerically demonstrate that $R^{\star\star}$ provides a close approximation to the optimal rate. Notably, this optimal rate is obtained by solving $D_{\Omega}(x) = 0$ defined in Appendix B, which involves lower computational complexity compared to an exhaustive search of finding maximum $\check{T}$. However, the precision of the approximation in~(\ref{happx}) depends on both $t$ and $B$, and any loss in accuracy may result in a deviation of the estimated $\check{T}$, thereby widening the optimality gap. We will also illustrate in Section~\ref{numr} how this gap changes under various parameter settings.
\end{remark}
\subsubsection{$x\in[\Lambda_0, \Lambda_1]$}
For remaining $x\in[\Lambda_0, \Lambda_1]$, we cannot determine the convexity-concavity or monotonicity of $\check{T}$ with respect to $x$, which we have to deploy exhaustive 1D grid search for the range~\cite{boyd, HBRIS}. Therefore, by the optimal $x^{\star\star\star}$ that maximizes $\check{T}$ in $[\Lambda_0, \Lambda_1]$ (i.e., $x^{\star\star\star}\triangleq \argmax_{x\in[\Lambda_0, \Lambda_1]} \check{T}$), the corresponding $R^{\star\star\star}$ is given by
\begin{equation}
\label{xr33}
R^{\star\star\star} = \log_2 (1+x^{\star\star\star}).
\end{equation}
\subsubsection{Finding Global Maximum}
By aggregating the three results, the optimal $R_o$ can be obtained as
\begin{equation}
\label{optR1}
R_o=\argmax_{R\in\{R^{\star}, R^{\star\star}, R^{\star\star\star}\}}T(R),
\end{equation}
and by incorporating the rate constraint, the final optimal solution $R^{\bullet}$ to Problem~(\ref{thpopt}) becomes
\begin{equation}
\label{opt47}
R^\bullet=\min\left(\max\left(R_{\min}, R_o\right), R_{\max}\right).
\end{equation}
The complete procedure for the proposed throughput maximization is summarized in Algorithm~\ref{euclid}. {It is worth emphasizing that Algorithm~\ref{euclid} is analytically grounded in the monotonic relationship between the outage probability and $|\eta|$ established in Theorem~\ref{thmnd}, rather than being a heuristic design.}

{The computational complexity of the proposed algorithm is dominated by three components: (i) a one-dimensional grid search with $M_{\Lambda}$ discrete candidates over $[\Lambda_0,\Lambda_1]$, (ii) a gradient-ascent procedure with $N_g$ iterations, and (iii) a Newton update with $N_n$ iterations. Accordingly, the overall computational complexity of Algorithm~\ref{euclid} scales as $\mathcal{O}\left((M_{\Lambda}+N_g+N_n)L\right)$, where $L$ is the number of nodes per dimension employed for integrand. The complexity grows linearly with the numerical resolution and the number of optimization iterations. Importantly, this complexity is independent of the system dimension. This stands in sharp contrast to brute-force search by searching $M_R$ rate candidate in $[R_{\min}, R_{\max}]$, where each candidate rate evaluation requires computing~\eqref{pout} involving $(N+2B)$-fold numerical integrals, resulting in a prohibitive complexity on the order of $\mathcal{O}(M_R L^{N+2B})$. By exploiting the analytical structure of the outage probability and confining the search to the reduced interval $[\Lambda_0,\Lambda_1]$ instead of whole given range of $R$, the proposed framework effectively converts a high-dimensional, exponential-complexity problem into a low-dimensional optimization with linear complexity. Furthermore, as numerically confirmed later in Fig.~\ref{graphT}, the effective operating region of the optimal rate occupies only a small fraction of the original brute-force search interval, which further reduces $M_{\Lambda}$ relative to $M_R$. As a result, the proposed algorithm achieves a favorable trade-off between accuracy and computational efficiency, making it well suited for practical and scalable FAS-ARIS systems, including scenarios that require real-time or online throughput optimization.}
		\begin{algorithm} [t]
	\caption{Proposed Throughput Maximization Algorithm}\label{euclid}
	\begin{algorithmic}[1]
		\State Determine $\Lambda_0 \triangleq \frac{1}{\bar{p}_1 U + \bar{p}_2} \left( \frac{1}{a} \psi^{-1}\left( \frac{1}{a^2} \right) \right)^2$
		\If {$\frac{\partial \check{T}}{\partial x}|_{x=\Lambda_0} \ge 0$}
		\State Determine $x^\star = \Lambda_0$
		\Else
		\State \multiline{Find $x^\star$ by solving $\frac{\partial\check{T}}{\partial x} =0$}
		\EndIf
		\State Determine $R^\star = \log_2 (1+x^\star)$
		\State Determine $\Lambda_1\triangleq\frac{1}{\bar{p}_2} \left(\frac{1}{a}\psi^{-1}\left(\frac{1}{a^2}\right)\right)^2$
		\State \multiline{Determine $\Omega \approx \mathbb{E}_{\Psi_p} [\sqrt{\bar{p}_1 s + \bar{p}_2}]$ and find $x_\omega$ by\\solving $D_\Omega (x)=0$ using Newton's Method}
		\State \multiline{Determine $x^{\star\star}=\max(\Lambda_1, x_\omega)$ and\\
		$R^{\star\star} = \log_2 (1+x^{\star\star})$}
		\State \multiline{Find  $x^{\star\star\star}\triangleq \argmax_{x\in[\Lambda_0, \Lambda_1]} \check{T}$ by grid search and determine $R^{\star\star\star} = \log_2 (1+x^{\star\star\star})$}
		\State Determine $R_o=\argmax_{R\in\{R^{\star}, R^{\star\star}, R^{\star\star\star}\}}T(R)$
		\State Determine $R^\bullet=\min\left(\max\left(R_{\min}, R_o\right), R_{\max}\right)$
	\end{algorithmic}
\end{algorithm}
{\section{Analysis-Guided Optimization Framework}
Our analysis establishes that by Theorem~\ref{thmnd}, $P_{\mathrm{out}}(t)$ is monotonically non-increasing in $t=|\eta|$. This implies that minimizing $P_{\mathrm{out}}$ with respect to $\boldsymbol{\Phi}$ is equivalent to maximizing $|\eta|$, which leads to the reformulation of~\eqref{outopt} into tractable~\eqref{outopt2}. The phase-alignment solution in~\eqref{optimalphi} is then obtained, which constitutes the first stage of the proposed framework.}

{Building on this analytical simplification, the throughput optimization problem in Section~\ref{thput} further exploits $\hat{P}_{\mathrm{out}}$, whose monotonic dependence on $t$ is also guaranteed in Remark~\ref{monohat}. By leveraging the analytical structure of $\hat{P}_{\mathrm{out}}$, we transform the original intractable problem into a one-dimensional rate optimization problem. Algorithm~\ref{euclid} systematically leverages the region-dependent behavior of $T(R)$, namely, its concave, quasiconcave, and non-monotonic characteristics, revealed by the analytical results in Section~\ref{thput}, thereby enabling near-optimal performance with low computational complexity.}

{Therefore, the role of the derived analysis is twofold: (i) it enables a closed-form ARIS phase optimization via monotonicity with respect to $|\eta|$, and (ii) it justifies the low-complexity, region-partitioned throughput optimization procedure in Algorithm~\ref{euclid}, which avoids exhaustive search while preserving near-optimality.}
\section{Numerical Results}\label{numr}
\subsection{Simulation Setup}
In our simulations, we adopt a FAS-ARIS system with the 3D coordinate of BS, ARIS, and MU, which are located at $(0, 0, 5)$ m, $(15, 15, 5)$ m, and $(55, 0, 0)$ m, respectively. The ARIS is placed close to the BS, as required by the system design condition in Section~\ref{4.1} that the ARIS should be located near the BS to guarantee a LoS BS-ARIS link. The channels between the ARIS and MU are modeled as flat Rician fading channels, with a path-loss exponent of 2.2 and $K=1$~\cite{FASRIS22, arisover}.
{The detailed simulation parameters are summarized in Table~\ref{SimPar}, which are chosen to reflect realistic operating conditions and practical hardware constraints in FAS- and ARIS-assisted systems~\cite{FASARIS, FASRIS22, FASRIS, aris5}. In particular, the ARIS size, power budget, and amplification gain follow the hardware-constrained ARIS modeling philosophy~\cite{FASARIS, aris5}, while the FAS-related parameters such as the number of ports, aperture size, and spatial correlation are consistent with commonly adopted FAS models in the literature~\cite{FASARIS, FASRIS22, FASRIS}. Low-to-moderate transmit power and target rate are considered to avoid artificially inflating performance~\cite{FASRIS22, FASRIS}, as to be depicted in Fig.~\ref{outP}(a) and~\ref{outP}(c), respectively. Overall, the selected parameters are aligned with prior FAS-RIS and FAS-ARIS studies; the reported performance trends are robust to reasonable variations of these values.} Monte Carlo simulation results are obtained by averaging over $10^4$ independent trials of signal transmission and reception conducted under the specified simulation environment. For benchmarking, we considered the FAS-RIS scenario as proposed in~\cite{FASRIS, FASRIS22}, which employs a conventional PRIS and the signal model is given as follows:
\begin{equation}
\label{bench}
\begin{cases}
y_k^{(p)} = \sqrt{P} \mathbf{h}_k^* \mathbf{\Phi} \mathbf{g} x + v_0~(k \in \mathbb{N})\\
{P}_{\mathrm{out}}^{(p)}=\mathbb{P}\left[\log_2 \left(1+\frac{P|A_{\max}|^2}{\sigma^2 }\right)<R\right]~(|A_{\max}|\triangleq\max_k A_k)\\
T=R(1-{P}_{\mathrm{out}}).
\end{cases}
\end{equation}
{We also evaluated the following canonical baselines:
\begin{itemize}
\item \textbf{Scenario without FAS:} It corresponds to a single-input single-output (SISO) ARIS-assisted systems, where the MU employs a single fixed antenna without spatial fluidity. In this case, $\mathbf{\Phi}$ is configured using the same phase-alignment principle described in Section~\ref{opop}.
\item \textbf{Perfect CSI scenario:} We assume that instantaneous full CSI is available for ARIS configuration, thereby providing a strict upper-bound on the achievable performance of the proposed FAS-ARIS system. Since $\{B_k\}$ is invariant with respect to $\mathbf{\Phi}$, minimizing $P_{\mathrm{out}}$ is equivalent to maximizing $\{|\mathbf h_k^*\mathbf{\Phi}\mathbf g|\}$ for the known $\{\mathbf h_k\}$ and $\{\mathbf g\}$. For a given $k$, this is achieved by $\theta_{m,k}^\star=\angle h_{k,m}-\angle g_m~(\forall m)$, which results in
$
|\mathbf h_k^*\mathbf{\Phi}^\star(k)\mathbf g|
=
\sum_{m=1}^M |h_{k,m}||g_m|.
$
Accordingly, the optimal MU port is selected as
\begin{equation}
\label{eq:k_opt_perfectCSI}
k^\star
\in
\arg\max_{k\in\mathcal N}
\sum_{m=1}^M |h_{k,m}||g_m|,
\end{equation}
with corresponding
$
\mathbf{\Phi}^\star
=
\mathrm{diag}
\big(
\{e^{j\theta_{m,k^\star}^\star}\}
\big).
$
\end{itemize}
For both benchmarks, $R^\bullet$ is determined via 1D exhaustive search over the feasible range of $R$.
}

\begin{center}
	\begin{table}[t] 
	\centering
		\caption{Simulation Parameters}
		\begin{tabular}{|>{\centering } m{1.2cm} |>{\centering} m{4.4cm} |>{\centering} m{1.4cm} | }
			\hline
			\textbf{Paramet{\tiny }er} & \textbf{Description} & \textbf{Value}
			\tabularnewline
			\hline
				\centering			$M$  & Number of ARIS elements\\(unless referred) & 4 \tabularnewline \hline
			\centering			$N$  & Number of FAS ports\\(unless referred) & 100  \tabularnewline \hline
			\centering                   $R$  & Threshold rate of ${P}_{\mathrm{out}}$\\(unless referred) & 2~[bps/Hz] \tabularnewline \hline
			\centering			$P$  & BS Transmit power\\(unless referred) & 10~[dBm]  \tabularnewline \hline
			\centering			$P'$  & Power budget of ARIS & Same as $P$  \tabularnewline \hline
			\centering			$\sigma^2$  & Thermal noise power & -104~[dBm] \tabularnewline \hline
			\centering			$\sigma_0^2$  & Dynamic noise power at ARIS & Same as $\sigma^2$  \tabularnewline \hline
			\centering			$W$  & Normalized FAS aperture\\with respect to $\lambda$ & 5  \tabularnewline \hline
			\centering		$\rho_{\max}^2$ & Maximum amplification gain of ARIS & 40~[dB]~\cite{aris5} \tabularnewline \hline
			\centering			$\mu^2$  & Correlation coefficient between\\FAS ports in same block & 0.97  \tabularnewline \hline
			\end{tabular}
		\label{SimPar}
	\end{table}
\end{center}
\begin{figure*}[t]
	\begin{center}
		\includegraphics[width=1.5\columnwidth,keepaspectratio]%
		{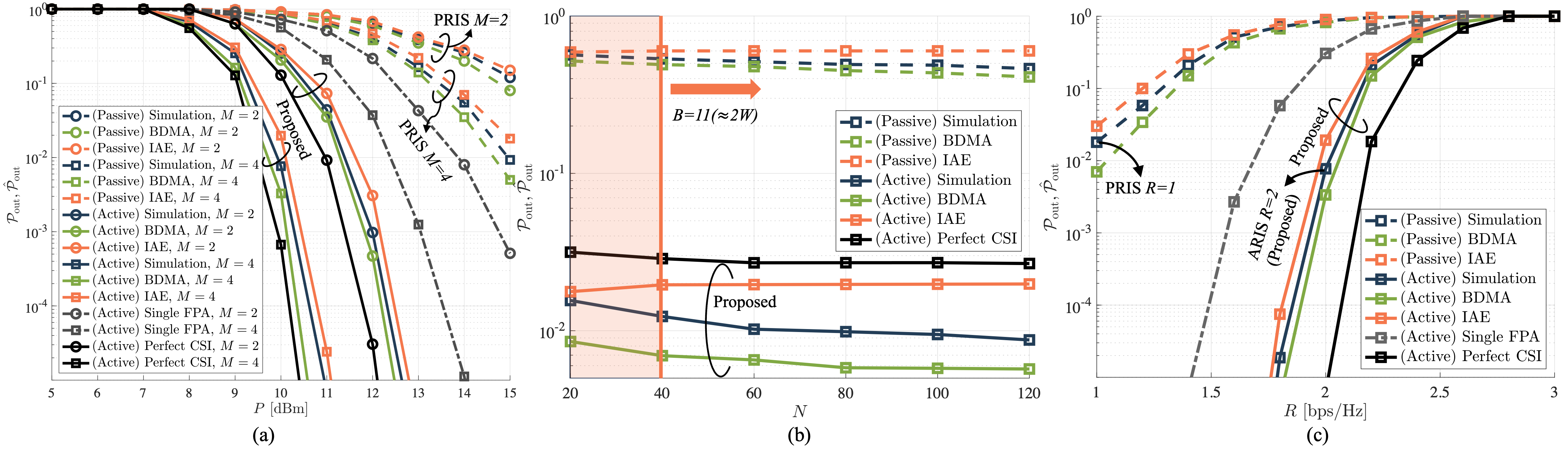}
		\caption{Outage probability with respect to the (a) $P$ with $M=2$ and $4$ (b) $N$ (c) $R$ under the implementation of FAS-ARIS/PRIS.}
		\label{outP}
	\end{center}
\end{figure*}
\subsection{Outage Probability Comparisons under Several Effects}
In Fig.~\ref{outP}, we examine the impact of transmit power $P$ on the outage probability ${P}_{\mathrm{out}}$ for different $M$. Each component in the legend refers to the following:
\begin{enumerate}
\item \textbf{``Simulation''}: the Monte Carlo-simulated outage probability obtained by averaging over a large number of independent trials
\item \textbf{``BDMA''}: theoretical analysis based on the BDMA model as derived in~(\ref{pout})
\item \textbf{``IAE''}: theoretical upper-bound obtained using the IAE approximation method described in~(\ref{poutupp})
{\item\textbf{``Single FPA''}: baseline with single fixed antenna at MU.
\item \textbf{``Perfect CSI''}: baseline assuming perfect CSI.}
\end{enumerate}
As observed in Fig.~\ref{outP}(a), the outage probabilities of all schemes decrease as $P$ increases. Additionally, for a given $P$, larger values of $M$ lead to lower outage probabilities across all schemes. This is because a larger $M$ provides more DoFs through both the ARIS/PRIS, thereby enhancing spatial diversity and improving link robustness~\cite{HBRIS, aris5}. Moreover, it is evident in Fig.~\ref{outP}(a) that both ``BDMA'' and ``IAE'' model closely matches the simulation results. 
Herein, the approximation accuracy of the ``IAE'' model is also comparable to the simulation results, supporting the use of the analytically tractable $\hat{{P}}_{\mathrm{out}}$ instead of ${P}_{\mathrm{out}}$ in the proposed throughput optimization framework, particularly when computational efficiency is prioritized. Importantly, the FAS-ARIS systems consistently outperforms the FAS-PRIS across the entire power range, and the performance gap becomes especially noticeable for $P > 9~\mathrm{dBm}$. This is due to the active amplification gain provided by ARIS, which enhances both the effective received power and the diversity gain via signal boosting, and validates the benefit of ARIS architecture in power-rich environments.

In Fig.~\ref{outP}(b), we illustrate the impact of the number of ports $N$ on the outage probability. As observed, the outage probabilities of all schemes generally decrease as $N$ increases, except for the ``IAE'' model, which exhibits a slower rate of decline. This is because a larger $N$ enhances the SNR of the MU through selection diversity, a trend accurately captured by both the ``BDMA'' model and the ``Simulation''. However, as $B$ approaches $2W$ with increasing $N$~\cite{newcor}, the diversity gain begins to saturate, resulting in a slower decline in outage probability. Notably, the ``IAE'' model, which is highly sensitive to the value of $B$ and invariant of block-diagonal structure, exhibits the most stable behavior with respect to $N$ under $W = 5$. In fact, the outage probability under the ``IAE'' model becomes nearly constant for $N>40$, as $B$ converges to $2W$~\cite{newcor}. 
Nonetheless, both ``BDMA'' and ``IAE'' model remain a computationally efficient and are sufficiently accurate for practical $N$, especially given the analytical intractability of direct simulation at high dimensions. Notably, ARIS consistently achieves lower outage probability than PRIS across the entire range of $N$, regardless of the spatial resolution. The performance superiority of FAS-ARIS over FAS-PRIS is persistent and robust, reinforcing its effectiveness in spatially correlated environments common in FAS-aided 6G network.
\begin{figure}[t]
	\begin{center}
		\includegraphics[width=0.6\columnwidth,keepaspectratio]%
		{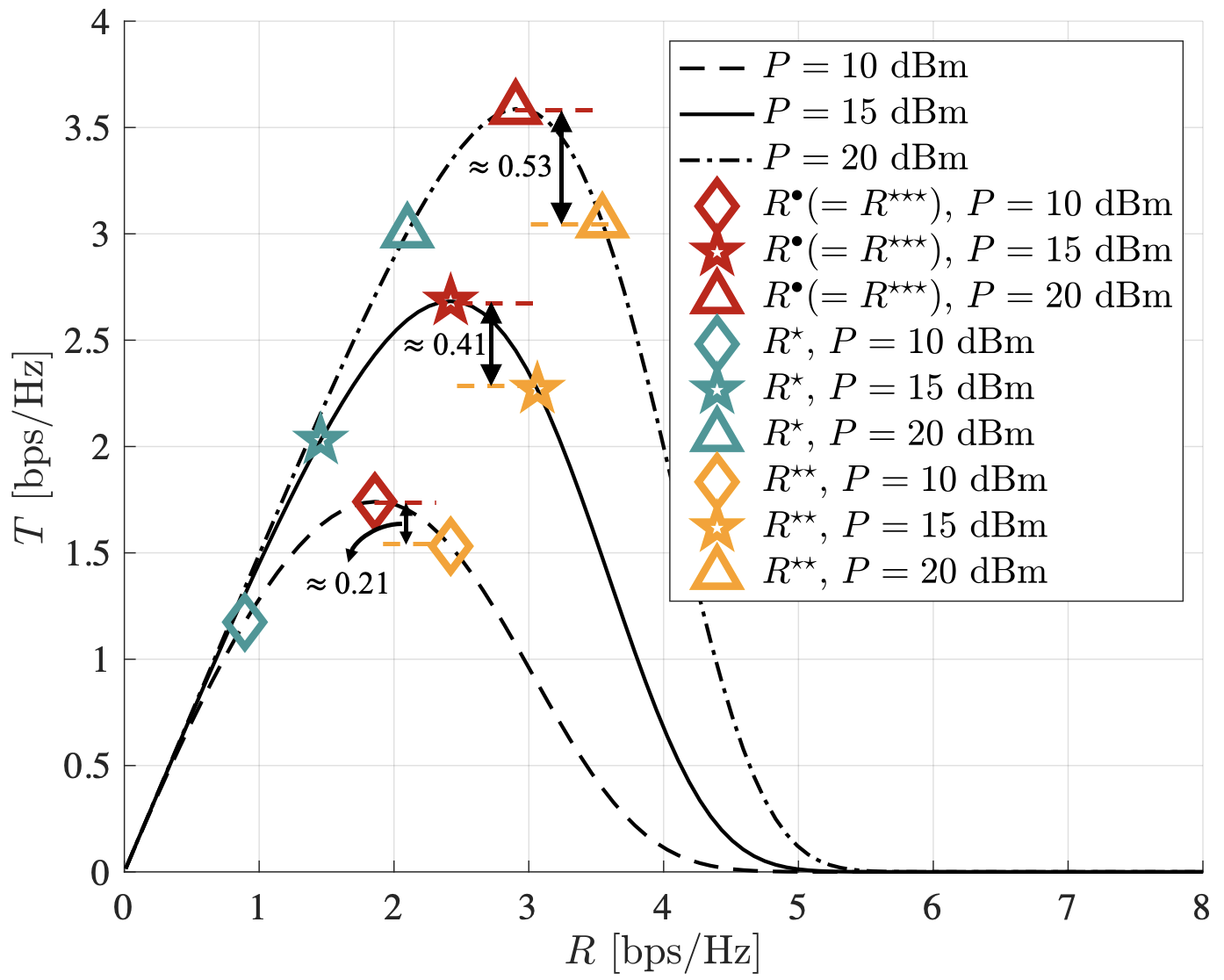}
		\caption{{Throughput with respect to $R$ under different $P$.}}
		\label{graphT}
	\end{center}
\end{figure}
\begin{figure*}[t]
	\begin{center}
		\includegraphics[width=1.4\columnwidth,keepaspectratio]%
		{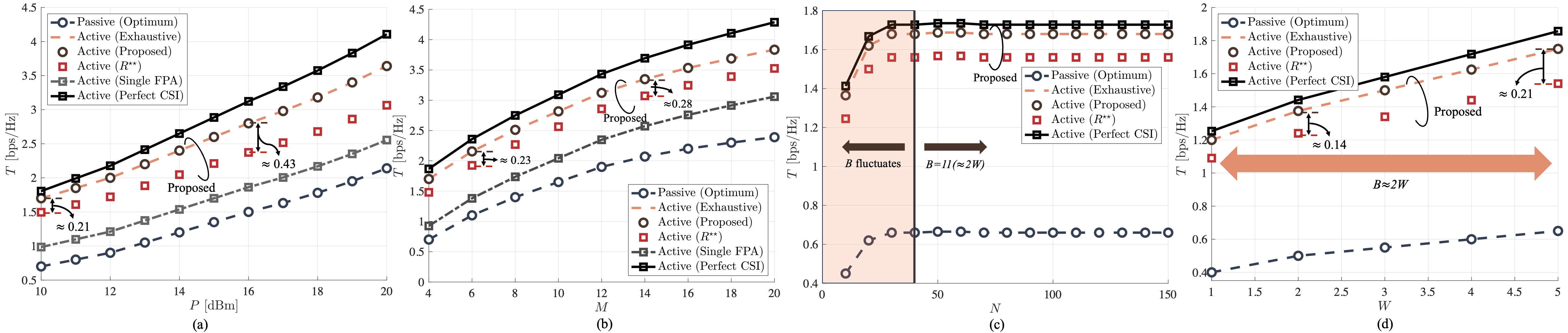}
		\caption{Throughput with respect to (a) $P$ (b) $M$ (c) $N$ (d) $W$ under the implementation of FAS-ARIS/PRIS.}
		\label{thP}
	\end{center}
\end{figure*}

Fig.~\ref{outP}(c) present the outage probability as a function of the threshold rate $R$. As observed, the outage probabilities of all schemes increase with $R$. The figure shows the comparable performance gap between ``Simulation'' and the ``IAE'' model and between ``Simulation'' and the ``BDMA'' model. This indicates that the ``IAE'' model provides an accurate approximation of the outage probability, further validating our approach of using~(\ref{poutupp}) for throughput maximization with respect to $R$. Furthermore, the FAS-ARIS configuration at $R = 2~\mathrm{bps/Hz}$ achieves even significantly lower outage probability than the FAS-PRIS system does even at $R = 1~\mathrm{bps/Hz}$. This highlights the performance superiority of FAS-ARIS in supporting higher data rates under the same reliability constraints, thanks to its active amplification gain and enhanced energy focusing capability. This performance gap becomes especially pronounced for $R < 2.4~\mathrm{bps/Hz}$, implying that in this regime, the advantage of FAS-ARIS, with its amplified reflection and increased effective DoFs, leads to a substantial reduction in outage events compared to FAS-PRIS, which lacks amplification gain. This makes FAS-ARIS a promising candidate for reliable and high-throughput FAS-aided communication systems in 6G.

{From Fig.~\ref{outP}, it is observed that the proposed FAS-ARIS scheme with statistical CSI and optimized ARIS phase configuration achieves $P_{\mathrm{out}}$ that closely approaches the ``Perfect CSI'' bound, while consistently outperforming the ``Single FPA'' baseline. This behavior is explained by the fact that, although only the long-term channel statistics is accessible, the derived phase optimization maximizes the effective deterministic channel components, thereby ensuring constructive signal combining at the selected FAS port. In contrast, the ``Single-FPA'' baseline completely disables spatial fluidity and port diversity at the receiver, which fundamentally limits its ability to exploit fading variations, leading to significantly higher $P_{\mathrm {out}}$. These results demonstrate that the proposed analysis and optimization on statistical-CSI-based FAS-ARIS design can effectively bridge the performance gap between practical implementations and the ideal full-CSI benchmark, without incurring the prohibitive overhead associated with instantaneous CSI acquisition.}

\subsection{Throughput Comparisons under Several Effects}
{In Fig.~\ref{graphT}, $T$ is plotted as a function of $R$ for $P=10, 15$ and $20~\mathrm{dBm}$, together with $R^{\star}$, $R^{\star\star}$, and $R^{\star\star\star}$. The curves clearly exhibit a unimodal behavior, reflecting the fundamental trade-off between rate and outage probability: while increasing $R$ initially improves throughput, $T$ sharply degrades beyond the peak due to the rapidly increasing outage probability, with a steeper decline observed at higher transmit powers. As indicated in the figure, $R^{\star\star\star}$ achieves the global maximum of $T$. This follows from the unimodality of $T$ over the interval $[\log_2(1+\Lambda_0),\log_2(1+\Lambda_1)]$ and the fact that
\begin{equation}
\label{rm3es}
\left.\frac{\partial \check{T}}{\partial x}\right|_{x=\Lambda_0}=\left.\frac{\partial \check{T}}{\partial R}\frac{\partial R}{\partial x}\right|_{x=\Lambda_0}\ge 0,
\end{equation}
as established in Remark~\ref{r000}, where $R^{\star}=\log_2(1+\Lambda_0)$. Herein, it further reveals that the effective optimization interval $[\log_2(1+\Lambda_0),\log_2(1+\Lambda_1)]$ is significantly narrower than the original brute-force search range; the size of the interval is approximately $1.61, 1.53$ and $1.45$ for $P=10, 15$ and $20~\mathrm{dBm}$, respectively, whereas the brute-force search range spans more than $R\in[0,6]$~bps/Hz. This clearly indicates that only a small portion of the full rate interval contributes to the optimal throughput. Therefore, Fig.~\ref{graphT} not only validates the unimodality of the throughput function, but also numerically confirms that the reduced interval $[\Lambda_0,\Lambda_1]$ is sufficiently tight in practice, which directly translates into a substantial reduction in computational complexity by shrinking the grid search from $M_R$ to $M_\Lambda$ in Section~\ref{thput}. Furthermore, $R^{\star\star}$ yields a throughput that closely approximates the globally optimal $R^{\star\star\star}$ without resorting to grid search. As shown in Fig.~\ref{graphT} and will be shown in Fig.~\ref{thP}, the resulting optimality gaps are small and explicitly quantified (approximately 0.21, 0.41 and 0.53 for $P=10, 15$ and $20~\mathrm{dBm}$, respectively). This confirms that Algorithm~\ref{euclid} provides a computationally efficient yet near-optimal solution for throughput maximization.}

In Fig.~\ref{thP}(a), we present the throughput performance across various transmit power values $P$. Each component in the legend refers to the following:
\begin{enumerate}
\item \textbf{``Passive''}: the FAS-PRIS systems as studied in~\cite{FASRIS22}
\begin{enumerate}
\item \textbf{``Passive (Optimum)''}: the throughput achieved by using a gradient ascent algorithm to find the optimum, as demonstrated in~\cite{FASRIS22}
\end{enumerate}
\item \textbf{``Active''}: the FAS-ARIS systems proposed in this work
\begin{enumerate}
\item \textbf{``Active (Proposed)''}: the proposed throughput optimization described in Section~\ref{thput}
\item \textbf{``Active (Exhaustive)''}:  the throughput obtained through exhaustive search over the 1D space of $R$
\item \textbf{``Active ($\mathbf{R}^{\star\star}$)''}: the proposed sub-optimal solution obtained from Fig.~\ref{graphT}, where we let $R^\bullet = R^{\star\star}$
{\item \textbf{``Active (Single FPA)''}: a baseline with single FPA at MU, thereby disabling spatial fluidity and resulting in consistently inferior performance compared to FAS-ARIS configurations.
\item \textbf{``Active (Perfect CSI)''}: an idealized performance benchmark assuming full CSI, which yields a strict performance upper-bound for given FAS-ARIS systems.}
\end{enumerate}
\end{enumerate}
From the figure, we observe that throughput improves with increasing $P$ across all configurations, where the proposed FAS-ARIS system can achieve the same throughput as FAS-PRIS while requiring significantly less power: FAS-PRIS needs around $P = 17 \sim 18~\mathrm{dBm}$ to match $T$ that FAS-ARIS achieves with only $P = 10~\mathrm{dBm}$. Notably, in the proposed FAS-ARIS system, ``Active (Proposed)'' model consistently achieves the same throughput as ``Active (Exhaustive)'' model, verifying the optimality of our proposed approach. The ``Active $(R^{\star\star})$'' model also performs well but is slightly sub-optimal, particularly at higher values of $P$. Moreover, the throughput gap between ``Active (Proposed)'' and ``Active $(R^{\star\star})$'' becomes more pronounced as $P$ increases: the optimality gap for $P = 10$ and $16~\mathrm{dBm}$ are approximately given by 0.21 and 0.43, respectively. This is because the term
\begin{equation}
\label{appintq}
\int_0^\infty Q_1 \left(a, \sqrt{(\bar{p}_1 s+\bar{p}_2)x}\right)\bar{g}(s)ds
\end{equation}
in~(\ref{TrxH}) becomes larger at higher $P$, leading to greater sensitivity to approximation errors in~(\ref{happx}). Consequently, the approximation in~(\ref{happx}) deviates more as $P$ increases, reducing the optimality of the throughput estimated by $R^{\star\star}$. {Nevertheless, even at higher transmit powers, the resulting optimality gap remains relatively small compared to the overall throughput values, thereby confirming the robustness and practical effectiveness of the proposed low-complexity solution. This trend is consistently observed in Fig.~\ref{thP}(b)-(d), where similar behaviors persist across different system parameters.}

In Fig.~\ref{thP}(b), we examine the throughput performance across various values of $M$. As observed, the throughput increases for all configurations in given range of $M$, consistent with the expected gain from additional reflecting elements. Similar to the behavior seen with increasing $P$, the ``Active (Proposed)'' model consistently achieves the same throughput as the ``Active (Exhaustive)'' model, demonstrating the optimality of our proposed framework. Furthermore, the proposed FAS-ARIS system can attain the same throughput as FAS-PRIS with substantially fewer elements: while FAS-PRIS requires approximately $M = 12 \sim 14$ to achieve a given $T$, FAS-ARIS reaches the same performance with just $M = 6$. Besides, similar to Fig.~\ref{thP}(a), the performance gap between the ``Active (Proposed)'' and the sub-optimal ``Active $(R^{\star\star})$'' model slightly increases with $M$: the optimality gap for $M = 6$ and $14$ are approximately given by 0.23 and 0.28, respectively. This is due to the following two effects in~(\ref{TrxH})~\cite{marQ}:
\begin{enumerate}
\item $a = \sqrt{2KM}$ in $Q_1$ increases with $M$, which leads to an increased value of $Q_1(a, b)$ for a fixed $b$.
\item As $M$ increases, the PDF $\bar{g}(s)$ becomes more right-skewed, assigning greater weight to larger values of $s$, which correspond to larger values of $b=\sqrt{(\bar{p}_1 s+\bar{p}_2)x}$. This shift tends to reduce $Q_1(a, b)$ with a fixed $a$, partially counteracting the effect of the increasing $a$.
\end{enumerate}
However, as shown in the figure, the impact of the increase in $a$ is slightly more dominant, resulting in a net increase of the integrand and hence the integral. 
As a result, the approximation used in~(\ref{happx}) becomes less accurate at larger $M$, leading to a degradation in the optimality of $R^{\star\star}$.

In Fig.~\ref{thP}(c), we analyze the throughput performance across various values of $N$. As shown in the figure, the throughput curves of all schemes exhibit fluctuations as $N$ varies from 10 to 40. This is primarily due to changes in the number of blocks $B$ in the diagonal correlation matrix, which affects the diversity structure. However, beyond $N=40$, the curves flatten out, indicating that $B$ becomes constant $(\approx 2W)$~\cite{newcor} and no longer changes with $N$. Importantly, since $N$ does not affect~(\ref{appintq}), the performance gap between the optimal and sub-optimal solution corresponds to $R^{\star\star}$ remains unchanged as $N$ increases. Thus, the accuracy of proposed approximation by $R^{\star\star}$ is preserved regardless of $N$ once $B$ stabilizes.

In Fig.~\ref{thP}(d), we investigate the throughput performance across various values of $W$. Since analytically characterizing the resulting optimality gap in closed form is mathematically intractable due to the non-linear dependence on $B$, we instead evaluate the gap numerically for representative values of $W$. As observed, the optimality gap between the ``Active (Proposed)'' and the sub-optimal ``Active $(R^{\star\star})$'' models increases with larger $W$. This behavior arises because, under a fixed $N$, increasing $W$ reduces the statistical dependency among the channels of each port, resulting in a larger $B$. Moreover, the convergence value of $B$, which is $2W$, also increases. This causes the approximation used in~(\ref{happx}) to become less accurate, since the first-order approximation $(1 - t)^B \approx 1 - Bt$ for fixed $t \ll 1$ becomes increasingly inaccurate as $B$ grows. Consequently, the degradation in approximation accuracy leads to a reduced optimality of $R^{\star\star}$: the optimality gap for $W = 2$ and $5$ are approximately given by 0.14 and 0.21, respectively.

{From Fig.~\ref{thP}, consistent with the trends observed in Fig.~\ref{outP}, the ``Active (Single FPA)'' and ``Active (Perfect CSI)'' curves respectively establish clear lower- and tight upper-bounds on $T$. The persistent gap underscores the fundamental importance of spatial reconfigurability and channel knowledge in shaping the throughput limits of FAS-ARIS systems, thereby delineating the practical performance region within which statistical-CSI-based designs are expected to operate. Together, these results quantitatively characterize the feasible performance bounds for practical FAS-ARIS implementations.}
\section{Conclusion}
In this paper, we developed an integrated FAS-ARIS framework to address the limitations of compact MIMO implementations in 6G. By explicitly modeling spatial correlation among fluid antenna ports and incorporating the effects of active noise amplification, we derived tractable outage probability expressions, provided a tight analytical upper-bound, and formulated efficient algorithms for outage and throughput optimization. The proposed region-partitioned optimization approach achieved near-optimal throughput with significantly reduced complexity, offering a practical trade-off between performance and computational cost. Numerical evaluations under diverse system settings consistently verified that FAS-ARIS outperforms conventional FAS-PRIS designs, achieving higher reliability and energy efficiency. 
These insights highlight FAS-ARIS as a promising building block for 6G, bridging theoretical advances with practical deployment considerations.
	\section*{Appendix A}
\section*{Proof of Theorem~\ref{quasi}}
It is clear that $\check{T}$ has a form of 
\begin{equation}
\label{TrxH}
\begin{aligned}
\check{T}
=&\log_2 (1+x)\\
&\left(1-\left(\int_0^\infty \left[1 - Q_1\left( a,\sqrt{(\bar{p}_1 s+\bar{p}_2)x}\right) \right]\bar{g}(s)ds\right)^B\right)\\
\mathop{=}^{(a)}&\log_2 (1+x)\\
&\left(1-\left(1-\int_0^\infty  Q_1\left( a,\sqrt{(\bar{p}_1 s+\bar{p}_2)x}\right) \bar{g}(s)ds\right)^B\right),
\end{aligned}
\end{equation}
where $a\triangleq\sqrt{2KM}$, $\bar{p}_1 \triangleq \frac{2\sigma^2}{P\bar{\sigma}^2}$ and $\bar{p}_2\triangleq \frac{2\sigma_0^2}{P\bar{\sigma}^2 \rho^{*2}}$. Moreover, $(a)$ holds since $\bar{g}(s)$ represents the PDF of a scaled non-central chi-squared distribution, as derived in~(\ref{pdfbb}).
\begin{figure}[t]
	\begin{center}
		\includegraphics[width=0.8\columnwidth,keepaspectratio]%
		{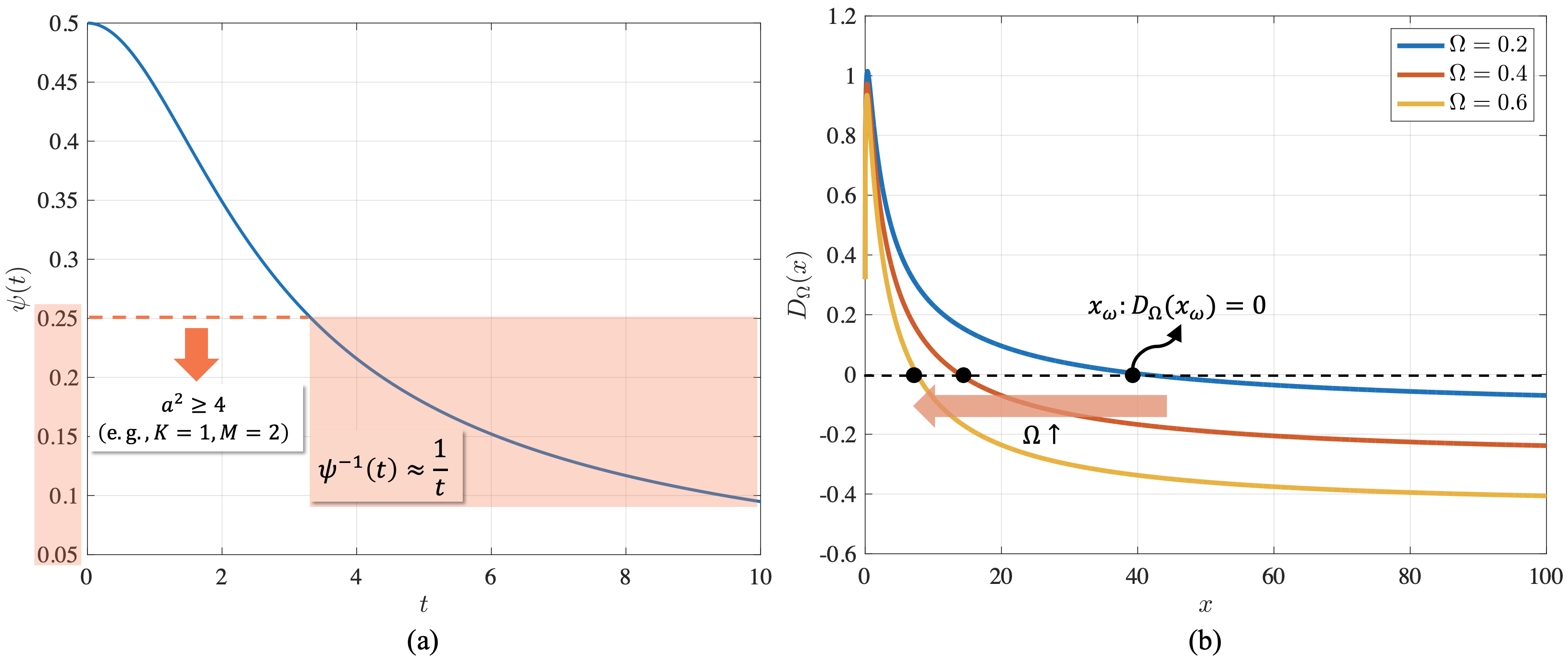}
		\caption{(a) Behavior of $\psi(t)\triangleq\frac{I_1 (t)}{tI_0 (t)}$. As the function value falls below 4, corresponding to $a^2 = 2KM = 4$ with, e.g., $K = 1$ and $M = 2$, we can conclude that $\psi(t)$ can be well-approximated by $\frac{1}{t}$. Furthermore, as $K$ and $M$ increase, this approximation becomes increasingly accurate. (b) Behavior of $D_{\Omega}(x)$ and $x_\omega$ for different $\Omega$ values. As $\Omega$ increases, the effect of $ - \frac{\Omega}{\sqrt{2\pi} x^{\frac{1}{4}}} \log_2(1 + x)$ in $D_{\Omega}$ increases, which leads to steeper decrease of $D_{\Omega}$. Furthermore, after $D_{\Omega}$ crosses zero, it converges to $0^-$.}
		\label{xphi}
		\end{center}
		\end{figure}
Now consider the first and second-order derivative of $\check{T}$, which can be derived in~(\ref{parx}).
\begin{figure*}
\begin{equation}
\label{parx}
\begin{cases}
\frac{\partial \check{T}}{\partial x}=&\frac{1}{(1+x)\ln 2}\left(1-\left(\int_0^\infty \left[1 - Q_1\left( a,\sqrt{(\bar{p}_1 s+\bar{p}_2)x}\right) \right]\bar{g}(s)ds\right)^B\right)\\
& - B\log_2(1+x) \left(\int_0^\infty \left[1 - Q_1\left( a,\sqrt{(\bar{p}_1 s+\bar{p}_2)x}\right) \right]\bar{g}(s)ds\right)^{B-1}\int_0^\infty\left( e^{- \frac{a^2 + (\bar{p}_1 s+\bar{p}_2)x}{2}} I_0\left( a \sqrt{ (\bar{p}_1 s+\bar{p}_2)x } \right) \frac{\bar{p}_1 s+\bar{p}_2}{2}\right) \bar{g}(s)ds\\
\frac{\partial^2 \check{T}}{\partial x^2}=& -\frac{1}{(1+x)^2 \ln 2} \left(1-\left(\int_0^\infty \left[1 - Q_1\left( a,\sqrt{(\bar{p}_1 s+\bar{p}_2)x}\right) \right]\bar{g}(s)ds\right)^B\right)\\
&-\frac{2}{(1+x)\ln 2}   B\left(\int_0^\infty \left[1 - Q_1\left( a,\sqrt{(\bar{p}_1 s+\bar{p}_2)x}\right) \right]\bar{g}(s)ds\right)^{B-1}\int_0^\infty\left( e^{- \frac{a^2 + (\bar{p}_1 s+\bar{p}_2)x}{2}} I_0\left( a \sqrt{ (\bar{p}_1 s+\bar{p}_2)x } \right)  \frac{\bar{p}_1 s+\bar{p}_2}{2}\right) \bar{g}(s)ds  \\
&+ \log_2(1 + x)  \Biggl[-B(B-1)(\int_0^\infty \left[1 - Q_1\left( a,\sqrt{(\bar{p}_1 s+\bar{p}_2)x}\right) \right]\bar{g}(s)ds)^{B-2} \\
&\left(\int_0^\infty\left( e^{- \frac{a^2 + (\bar{p}_1 s+\bar{p}_2)x}{2}} I_0\left( a \sqrt{ (\bar{p}_1 s+\bar{p}_2)x } \right)  \frac{\bar{p}_1 s+\bar{p}_2}{2}\right) \bar{g}(s)ds \right)^2-B(\int_0^\infty \left[1 - Q_1\left( a,\sqrt{(\bar{p}_1 s+\bar{p}_2)x}\right) \right]\bar{g}(s)ds)^{B-1}\\
 &\underbrace{\int_0^\infty}_{(\dagger)} \underbrace{e^{-\frac{a^2 +(\bar{p}_1 s+\bar{p}_2)x}{2}} \left(-I_0(a\sqrt{(\bar{p}_1 s+\bar{p}_2)x})   \frac{\bar{p}_1 s+\bar{p}_2}{2} + a I_1(a\sqrt{(\bar{p}_1 s+\bar{p}_2)x})  \frac{ \bar{p}_1 s+\bar{p}_2 }{2 \sqrt{ (\bar{p}_1 s+\bar{p}_2)x} } \right) \frac{\bar{p}_1 s+\bar{p}_2}{2}\bar{g}(s)}_{\triangleq o(s, x)}ds\Biggr]
\end{cases}
\end{equation}
\hrule
\end{figure*}
Here, $\frac{\partial \check{T}}{\partial x}\le 0$ holds, and we know that $\frac{\partial^2 \check{T}}{\partial x^2}\le 0$ when
\begin{equation}
\label{ineqstar}
-I_0(a\sqrt{(\bar{p}_1 s+\bar{p}_2)x})  +   \frac{ a I_1(a\sqrt{(\bar{p}_1 s+\bar{p}_2)x}) }{ \sqrt{ (\bar{p}_1 s+\bar{p}_2)x} } \ge 0,
\end{equation}
which can be transformed into 
\begin{equation}
\label{ge0}
\begin{aligned}
x \le \frac{1}{\bar{p}_1s + \bar{p}_2}\left(\frac{1}{a} \psi ^{-1}\left(\frac{1}{a^2}\right)\right)^2\triangleq\Lambda(s),
\end{aligned}
\end{equation}
where the monotonically decreasing behavior of $\psi(t)$ ensures the existence of $\psi^{-1}$. Moreover, since $a^2 = 2KM$ is assumed to be sufficiently large to carry out the following approximation of $\psi(t)\approx \frac{1}{t}$~\cite{bessbook}, which is numerically shown in Fig.~\ref{xphi}(a), we can conclude that
\begin{equation}
\label{psiapp}
\begin{aligned}
\psi(t) \approx \frac{1}{t}&\rightarrow \psi^{-1}(t) \approx \frac{1}{t}\rightarrow \left( \frac{1}{a} \psi^{-1}\left( \frac{1}{a^2} \right) \right)^2 \approx a^2.
\end{aligned}
\end{equation}
To determine the range of $x$ that ensures the concavity of $\check{T}$, we approximate the upper-limit of $(\dagger)$ by a large but finite constant $U \gg 1$ instead of $\infty$. Then, for any $x$ that satisfies~(\ref{ge0}) and $\forall s \in [0, U]$,~(\ref{ineqstar}) holds throughout this interval, and $\frac{\partial^2 \check{T}}{\partial x^2} \le 0$ is thereby guaranteed. Accordingly, we define
\begin{equation}
\label{lambdam}
\Lambda_0 \triangleq \Lambda(U) = \frac{1}{\bar{p}_1 U + \bar{p}_2} \left( \frac{1}{a} \psi^{-1} \left( \frac{1}{a^2} \right) \right)^2,
\end{equation}
which is obtained by substituting $s = U$ into $\Lambda(s)$ in~(\ref{ge0}). We can therefore conclude that $\check{T}(x)$ is concave, i.e., $\frac{\partial^2 \check{T}}{\partial x^2} < 0$, for all $x \le \Lambda_0$, and the theorem follows.~$\blacksquare$
\begin{remark}
\label{rem1}
As $s\rightarrow\infty$, $o(s,x)$ goes to 0, since $I_0 (t)\approx I_1(t)$ for $t=a\sqrt{(\bar{p}_1 s+\bar{p}_2)x}\rightarrow\infty$~\cite{bessbook}, and then $o(s,x)$ becomes
\begin{equation}
\label{osx}
\begin{aligned}
&o(s,x)\sim \mathcal{O} \left(s^2 e^{-s} I_0 \left(\sqrt{s}\right)\right)\rightarrow0~\left(\because I_0 (\sqrt{s}) \sim \mathcal{O}\left(s^{-\frac{1}{4}}e^{\sqrt{s}}\right)\right),
\end{aligned}
\end{equation}
which implies that we can always take $U$ to sufficiently suppress $\int_U^\infty o(s,x)ds$ regardless of any $x \ge 0$.
\end{remark}
	\section*{Appendix B}
\section*{Proof of Theorem~\ref{max22}}
		By approximation in~(\ref{psiapp}), $\Lambda_1\approx \frac{a^2}{\bar{p}_2}$, which implies that
	\begin{equation}
		\label{imp}
		\begin{aligned}
		&x \ge \Lambda_1\leftrightarrow \bar{p}_2 x \ge a^2  
		\leftrightarrow \sqrt{(\bar{p}_1 s+\bar{p}_2 )x} \ge a~(\forall s\ge 0).
		\end{aligned}
		\end{equation}		
Under~(\ref{imp}), we can asymptotically approximate $Q_1$ by~\cite{marQ}
\begin{equation}
\label{qq}
\begin{aligned}
Q_1 (a, \sqrt{ (\bar{p}_1 s+\bar{p}_2)x})&\approx \sqrt{\frac{\sqrt{ (\bar{p}_1 s+\bar{p}_2)x}}{a}}Q\left(\sqrt{ (\bar{p}_1 s+\bar{p}_2)x}-a\right)\\
&\mathop{\approx}^{(b)} \frac{1}{2}\sqrt{\frac{\sqrt{ (\bar{p}_1 s+\bar{p}_2)x}}{a}}e^{-\frac{(\sqrt{ (\bar{p}_1 s+\bar{p}_2)x}-a)^2}{2}}
\end{aligned}
\end{equation}
where $Q(\cdot)$ is the Gaussian $Q$-function and $(b)$ holds by $Q(x)\approx \frac{1}{2} e^{-\frac{x^2}{2}}$ for $x\ge0$~\cite{bessbook}. Hence, $\check{T}$ becomes
\begin{equation}
\label{tqtq}
\begin{aligned}
\check{T}\approx& B x^{\frac{1}{4}}\log_2 (1+x)\\
&\int_0^\infty \sqrt{\frac{\sqrt{\bar{p}_1 s+\bar{p}_2}}{a}} Q\left(\sqrt{ (\bar{p}_1 s+\bar{p}_2)x}-a\right)\bar{g}(s)ds.
\end{aligned}
\end{equation}
Now define:
\begin{equation}
\label{i1i2}
\begin{cases}
W_1(x) \triangleq &\int_0^\infty \frac{1}{2} \sqrt{ \frac{ \sqrt{ \bar{p}_1 s + \bar{p}_2 } }{a} } e^{ - \frac{ ( \sqrt{(\bar{p}_1 s + \bar{p}_2)x} - a )^2 }{2} } \bar{g}(s) ds, \\
W_2(x) \triangleq &\frac{1}{\sqrt{2\pi}} \frac{1}{2}\int_0^\infty \sqrt{ \frac{ \sqrt{ \bar{p}_1 s + \bar{p}_2 } }{a} }\sqrt{ \bar{p}_1 s + \bar{p}_2 }\\
&~~~~~~~~~~~~~ e^{ - \frac{ ( \sqrt{(\bar{p}_1 s + \bar{p}_2)x} - a )^2 }{2}} \bar{g}(s)ds.
\end{cases}
\end{equation}
Then the derivative of $\check{T}$ is approximately:
\begin{equation}
\label{chedapp}
\begin{aligned}
\frac{\partial \check{T}}{\partial x} \approx & B \Biggl\{ \left( \frac{x^{\frac{1}{4}}}{(1 + x)\ln 2} + \frac{1}{4} x^{-\frac{3}{4}} \log_2(1 + x) \right) W_1(x) \\
& ~~~~- x^{-\frac{1}{4}} \log_2(1 + x)   W_2(x) \Biggr\}.
\end{aligned}
\end{equation}
Note that both $W_1(x)$ and $W_2(x)$ share $\Psi$:
\begin{equation}
\label{defpsi}
\Psi(s, x) \triangleq  \frac{1}{2} \sqrt{ \frac{ \sqrt{ \bar{p}_1 s + \bar{p}_2 } }{a} }  e^{ - \frac{ ( \sqrt{(\bar{p}_1 s + \bar{p}_2)x} - a )^2 }{2}}\bar{g}(s).
\end{equation}
Then, the integrals in~(\ref{i1i2}) can be rewritten as:
\begin{equation}
\label{i1i2two}
\begin{cases}
W_1(x) = \int_0^\infty \Psi(s, x)ds, \\
W_2(x) = \frac{1}{\sqrt{2\pi }}  \int_0^\infty \sqrt{\bar{p}_1 s + \bar{p}_2} \Psi(s, x)ds,
\end{cases}
\end{equation}
which implies that their ratio becomes a weighted mean:
\begin{equation}
\label{iiwm}
\sqrt{2\pi} \frac{W_2(x)}{W_1(x)} =  \mathbb{E}_{\Psi_p} [\sqrt{\bar{p}_1 s + \bar{p}_2}]\triangleq\Omega,
\end{equation}
where $\Psi_p\triangleq \frac{\Psi(s, x)}{\int_0^\infty \Psi(s,x)ds}$ is a normalized PDF of $\Psi(s,x)$ with respect to $s$. 
Thereafter, 
we can approximate~(\ref{chedapp}) by:
\begin{equation}
\label{deri2}\
\frac{\partial \check{T}}{\partial x} \approx B  W_1(x) D_{\Omega}(x),
\end{equation}
where $D_{\Omega}(x)\triangleq \frac{x^{\frac{1}{4}}}{(1 + x)\ln 2} + \frac{1}{4} x^{-\frac{3}{4}} \log_2(1 + x) - \frac{\Omega}{\sqrt{2\pi} x^{\frac{1}{4}}} \log_2(1 + x).$
Since $B W_1(x) > 0$ for all $x > 0$, the sign of~(\ref{deri2}) is governed by $D_{\Omega} (x)$. It is clear that
\begin{enumerate}
\item $x\ll 1$: $D_{\Omega} (x) \sim \frac{5}{4\ln 2} x^{\frac{1}{4}} (\approx 0^+)$
\item $x\gg 1$: $D_{\Omega} (x) \sim -\frac{\Omega}{\sqrt{2\pi} x^{\frac{1}{4}}} \log_2 x\rightarrow 0^-$
\end{enumerate}
In addition, from Fig.~\ref{xphi}(b), $D_{\Omega}$ crosses zero exactly once by $0^+\rightarrow 0^-$, which suggests the existence of a unique maximizer $x_\omega$ such that $D_{\Omega}(x_\omega) = 0$ and leads to $\frac{\partial \check{T}}{\partial x}|_{x=x_\Omega} = 0$. Note that $x_\omega$ can be easily found by Newton's Method~\cite{boyd}. This implies the followings:
\begin{enumerate}
\item $\Lambda_1 < x_\omega$: The global maximizer is uniquely determined, which is $x_\omega$, confirming the quasiconcavity of $\check{T}$ within the specified range~\cite{boyd}. 
\item $\Lambda_1 \ge x_\omega$: $D_{\Omega}(x) < 0~(\forall x \ge \Lambda_1)$, indicating that $\check{T}$ is a monotonically decreasing function in this region.
\end{enumerate} 
This concludes the proof and the theorem follows.~$\blacksquare$

		\bibliographystyle{IEEEtran}
		\bibliography{IEEEexample}	
\end{document}